\documentclass[reqno,11pt]{article}

\usepackage[english]{babel}
\usepackage[utf8]{inputenc}
\usepackage[T1]{fontenc}
\usepackage{fec}
\usepackage[a4paper, total={6.5in, 8.6in}]{geometry}
\usepackage[ruled, algo2e]{algorithm2e} 
\usepackage{mathrsfs}

\usepackage{mathtools}
\usepackage{amssymb,amsmath,amsthm,amsfonts}
\usepackage{graphicx}
\usepackage{booktabs}
\usepackage{todonotes}
\newtheorem{theorem}{Theorem}
\newtheorem{definition}{Definition}
\newtheorem{lemma}{Lemma}
\newtheorem{corollary}{Corollary}
\newtheorem{proposition}{Proposition}

\newtheorem{assumption}{Assumption}
\newtheorem{remark}{Remark}

\usepackage{xcolor}

\usepackage{authblk}

\newcommand{\bra}[1]{\left\langle{#1}\right|}
\newcommand{\ket}[1]{\left|{#1}\right\rangle}

\DeclarePairedDelimiterX\braket[2]{\langle}{\rangle}{#1 \delimsize\vert #2}
\DeclarePairedDelimiterX\ketbra[2]{| }{|}{#1 \delimsize\rangle\!\delimsize\langle #2}

\DeclareMathOperator*{\argmin}{arg\,min}

\DeclareMathOperator{\polylog}{polylog}
\DeclareMathOperator{\Tr}{Tr}

\newcommand{\td}{\ensuremath{\text{d}}}
\newcommand{\nnz}{\ensuremath{\textup{\textsf{nnz}}}}
\renewcommand{\d}{\mathrm{d}}

\usepackage{pifont}
\newcommand{\cmark}{\ding{51}}
\newcommand{\xmark}{\ding{55}}

\newcommand{\bfi}{\ensuremath{\boldsymbol{i}}}

\title{A quantum central path algorithm for linear optimization}

\author[1]{Brandon Augustino} 
\author[2,3]{Jiaqi Leng\thanks{Corresponding Author: \texttt{jiaqil@vt.edu}}}
\author[4]{Jialu Bao} 
\author[4,5]{Giacomo Nannicini}
\author[6]{Tam\'as~Terlaky}
\author[7,8]{Xiaodi Wu}

\affil[1]{\small{\textit{Sloan School of Management, Massachusetts Institute of Technology}}}
\affil[2]{\textit{Simons Institute and Department of Mathematics, University of California, Berkeley}}
\affil[3]{\textit{Institute for Advanced Computing and Department of Computer Science, Virginia Tech}}
\affil[4]{\textit{Department of Industrial and Systems Engineering, University of Southern California}}
\affil[5]{\textit{Department of Electrical and Computer Engineering, University of Southern California}}
\affil[6]{\textit{Department of Industrial and Systems Engineering, Lehigh University}}
\affil[7]{\textit{Department of Computer Science, University of Maryland, College Park}}
\affil[8]{\textit{Joint Center for Quantum Information and Computer Science, University of Maryland}}

\date{\today}

\begin{document}
\maketitle

\begin{abstract}
We propose a novel quantum algorithm for solving linear optimization problems by quantum-mechanical simulation of the central path. While interior point methods follow the central path with an iterative algorithm that works with successive linearizations of the perturbed KKT conditions, we perform a simulation working directly with the nonlinear complementarity equations. Using a two-stage quantum eigenpath traversal schedule, this approach yields an algorithm for solving linear optimization problems involving $m$ constraints and $n$ variables to $\varepsilon$-optimality using $\widetilde{\Ocal}\!\left((m+n)^{3/2}\,\varepsilon^{-1}+(m+n)^{1/2}\,\varepsilon^{-3/2}\right)$ queries to an oracle that evaluates a potential function. In the standard gate model (i.e., without access to quantum RAM), our algorithm can obtain $\varepsilon$-precise solutions to LO problems using at most
$$\widetilde{\Ocal}\!\left(\left[\frac{(m+n)^{3/2}}{\varepsilon} + \frac{(m+n)^{1/2}}{\varepsilon^{3/2}}\right] \cdot \nnz(A)\right)$$ 
elementary gates, where $ \nnz (A)$ is the total number of non-zero elements found in the constraint matrix. 
\end{abstract}

\newpage 
\tableofcontents
\newpage

\section{Introduction}\label{sec:Intro}
Given a matrix $A \in \R{m \times n}$ and vectors $b \in \R{m}$ and $c \in \R{n}$, we are interested in solving the following Linear Optimization (LO) problem:
\begin{equation}\label{e:LP} 
    \begin{aligned}
        \min_{x \in \R{n}}\quad & c^{\top} x \\
        \text{subject to}\quad &A x      \geq b, \\
        & ~~x \geq 0,
    \end{aligned}
\end{equation}
where $x \in \R{n}$ is the \textit{primal} variable. We assume $\| A \|, \|b \|, \| c \| \leq 1$ for normalization purposes. Normalization of the input data is a common assumption in the literature. The primal problem \eqref{e:LP} has an associated \textit{dual} problem, given by
\begin{equation}\label{e:LP-D} 
    \begin{aligned}
        \max_{y \in \R{m}}\quad & b^{\top} y \\
        \text{subject to}\quad &A^{\top} y  \leq c, \\
        &\quad ~y \geq 0, 
    \end{aligned}
\end{equation}
where $A^{\top} \in \R{n\times m}$ and $y \in \R{m}$ is the dual variable. Whenever the \textit{canonical} primal and dual LO problems \eqref{e:LP}-\eqref{e:LP-D} are feasible, strong duality holds: there exists a primal-dual optimal solution $(x_*, y_*)$ to \eqref{e:LP}-\eqref{e:LP-D} with vanishing duality gap, i.e., $c^{\top} x_* - b^{\top} y_* = 0$.

Linear Optimization is routinely used to model and solve fundamental problems in economics, finance, and engineering. LO was brought to the forefront of computer science and applied mathematics following the advent of digital computer along with Dantzig's Simplex method \cite{dantzig1948programming, dantzig1955generalized}. The first polynomial time algorithm for LO was Kachiyan's ellipsoid method \cite{khachiyan1980polynomial}, though it failed to elicit an efficient practical implementation. Shortly thereafter, Karmarkar introduced his projective method for LO problems in the seminal work \cite{karmarkar1984new}, which improved the complexity of the ellipsoid method. Although interior point algorithms had been studied because at least the 1950's \cite{dikin1967iterative, dikin1974convergence, fiacco1964sequential, fiacco1990nonlinear, frisch1954principles, frisch1955logarithmic, frisch1956resolution}, Karmarkar's algorithm was the first to run in polynomial time, and the so-called Interior Point Method (IPM) revolution was underway. 

Enhancements over Karmarkar's algorithm followed closely; Renegar \cite{renegar1988polynomial} improved the iteration complexity,  Vaidya \cite{vaidya1987algorithm} and Gonzaga \cite{gonzaga1989algorithm} simultaneously reduced the overall complexity, and Nesterov and Nemirovskii \cite{nesterov1988general, nesterov1994interior} introduced the paradigm of \textit{self-concordant barrier functions}. Today, the IPM literature constitutes a vast and celebrated line of research. We do not provide a comprehensive review here; rather, we point the reader to the excellent texts \cite{renegar2001mathematical, roos2005interior, terlaky2013interior, wright1997primal, ye2011interior} and the references therein. Many recent developments can be found in \cite{van2020deterministic, van2020solving, cohen2021solving, jiang2020faster, jiang2020fasterLP, lee2014path, lee2015efficient, vladu2023interior}, and the current state of the art running time results are obtained by the randomized IPM of Cohen, Lee and Song~\cite{cohen2021solving} and the deterministic IPM from van den Brand~\cite{van2020deterministic}. These IPMs can solve LO problems to precision $\varepsilon \in (0,1)$ in $\widetilde{\Ocal}_{1/\varepsilon}((m + n)^{\omega})$ time\footnote{The authors in \cite{van2020deterministic, cohen2021solving} solve the standard form LO problem which requires the assumption $\rank (A) = m$ (and in particular, $m \leq n$), so their complexity result is typically written as $\widetilde{\Ocal}_{1/\varepsilon}(n^{\omega})$. We make the dependence on $m$ explicit because this assumption is unnecessary in the canonical (or, symmetric) form we consider: introducing $m$ slack variables transforms the canonical problem into standard form, and the augmented coefficient matrix is of full row-rank. We also remark that the complexity of the algorithm in \cite{cohen2021solving} is more accurately expressed as $\widetilde{\Ocal}_{1/\varepsilon} \left(n^{\omega} + n^{2.5 - \frac{\alpha}{2}} + n^{2 + \frac{1}{6}} \right)$, where $\alpha$ is the dual-exponent of $\omega$. This simplifies to $\widetilde{\Ocal}_{1/\varepsilon} \left(n^{\omega} \right)$ for the current value $\omega \approx 2.38$. Jiang et al.~\cite{jiang2020fasterLP} improved the $n^{2 + \frac{1}{6}}$ term to $n^{2 + \frac{1}{18}}$.}, where $\omega \in [2, 2.38)$ is the exponent of the running time for matrix multiplication (i.e., we can multiply two $n \times n$ matrices in time $\Ocal(n^{\omega}))$, and the notation $\widetilde{\Ocal}_{\alpha, \beta} (f(x))$ suppresses polylogarithmic factors in $f(x)$, $\alpha$ and $\beta$ appearing in the overall running time (a formal definition of this notation is given in Section~\ref{s:Prelim}). 

One can view IPMs as a homotopy approach for Newton's method. IPMs are initialized to a strictly feasible point $x^{(0)} \in \{ x \in \R{n} : Ax > b,~x > 0 \}$, and approximately track an analytic curve known as the \textit{central path} towards optimality. Specifically, letting $\mu > 0$ and defining 
\[\phi : \{ x \in \R{n} : Ax > b,~ x>0 \} \mapsto \R{},\]
a generic path-following scheme solves a sequence of \textit{barrier problems} of the form 
\begin{equation}\label{e:barrierMU}
        x(\mu) \coloneqq \argmin_{\{ x \in \R{n} : Ax > b,~ x>0 \}}\quad   c^{\top} x + \mu \phi (x),
\end{equation}
using Newton's method. The \textit{barrier function} $\phi$ is chosen to be \textit{self-concordant}: the value $\phi(x)$ diverges to $\infty$ upon approaching the boundary of $\{ x \in \R{n} : Ax > b,~ x>0 \}$, and at a high level, the norms of higher-order derivatives of $\phi$ can be bounded in terms of its Hessian (the second-order Taylor series approximation of \eqref{e:barrierMU} is highly accurate). The value of $\mu$ is decreased in each iteration, and the optimal solution is reached by tracing the central path $\{ x (\mu) : \mu > 0\}$ as $\mu \to 0$. The use of the self-concordant barrier function $\phi$ in \eqref{e:barrierMU} ensures that we are always in the region of rapid local convergence enjoyed by Newton's method. 

It is also standard in the literature to define the central path as the set of minimizers associated with
\begin{equation}\label{e:barrier}
      x(t) \coloneqq   \argmin_{\{ x \in \R{n} : Ax > b,~ x>0 \}}\quad t c^{\top} x + \phi (x).
\end{equation}
The value of $t$ is increased in each iteration, and so tracing the central path $\{ x (t) : t > 0 \}$, we approach an optimal solution to \eqref{e:LP} as $t \to \infty$. Note however, there is no meaningful difference between \eqref{e:barrierMU} and \eqref{e:barrier}. Indeed, taking $\mu = 1/t$ and applying Newton's method to \eqref{e:barrierMU} yields the same sequence of minimizers as one would obtain from applying Newton's method to \eqref{e:barrier}. Yet, the perspective offered by \eqref{e:barrierMU} can be attractive for both theoretical and practical reasons; tracking $t \to \infty$ in \eqref{e:barrier} would require computation involving huge numbers. 

In every iteration of the classical IPM, one must obtain the solution to a linearized set of perturbed Karush-Kuhn-Tucker (KKT) optimality conditions known as the Newton linear system. The computation of this so-called \textit{Newton step} is usually performed by solving dense linear systems of equations on a very large scale. The IPMs that achieve the best complexity results \cite{van2020deterministic, cohen2021solving} perform this step exactly in the first iteration, and subsequently utilize sophisticated data structures to maintain an approximation of the Hessian inverse, amortizing the cost of solving the linear systems over the run of the algorithm. To the best of our knowledge, the improved theoretical guarantees of the IPMs in \cite{van2020deterministic, cohen2021solving} have not led to a practical implementation at the time of writing. Additionally, while the previous decade saw improvements in the iteration complexity for the specific applications of maximum flow \cite{madry2013navigating, madry2016computing, agarwal2024parallel}, minimum cost flow \cite{cohen2017negative}, matrix scaling \cite{cohen2017matrix}, and $\ell_p$-regression \cite{bubeck2018homotopy}, whether the asymptotic iteration count of $\Ocal (\sqrt{n} \log (1/\varepsilon))$ (or, $\Ocal (\sqrt{m} \log (1/\varepsilon))$ if $m \ll n$) can be reduced in the general setting remains an important open question in optimization. Quantum IPMs (QIPMs) currently found in the literature seek to accelerate the solution of the Newton linear system via quantum subroutines, hence they do not reduce the iteration count. We review these works in detail next. 

\subsection{Related work}
QIPMs were first introduced by Kerenidis and Prakash \cite{kerenidis2020quantum}, who proposed a quantum algorithm for solving LO and Semidefinite Optimization (SDO) problems. The main idea of their approach is to solve the Newton linear system at each iterate using a quantum linear systems algorithm (QLSA) \cite{harrow2009quantum, childs2012hamiltonian, chakraborty2018power}, and obtain a classical estimate of the resulting quantum state via quantum state tomography. This approach limits one to inexactly solving the Newton system at each iterate, and therefore, additional safeguards need to be taken in order to guarantee convergence. To reconcile the quantum noise introduced into the Newton steps, Augustino~et~al.~\cite{augustino2021quantum} proposed two convergent QIPMs for SDO and LO. The first closely quantized the classical Inexact-Infeasible IPM of \cite{toh2002solving}. The second framework is a novel Inexact-Feasible QIPM (IF-QIPM) that uses a \textit{nullspace representation} of the Newton system to ensure the sequence of iterates maintain primal-dual feasibility, in spite of using an inexact linear system subroutine. These ideas were specialized to LO by Mohammadisiahroudi et al.~\cite{mohammadisiahroudi2022efficient, mohammadisiahroudi2023inexact,mohammadisiahroudi2025quantum}, who use the iterative refinement algorithm for LO from Gleixner et al.~\cite{gleixner2012improving, gleixner2016iterative, gleixner2020linear} to exponentially improve the dependence on $\varepsilon^{-1}$, the inverse precision to which we seek to solve the primal and dual LO problems \eqref{e:LP}-\eqref{e:LP-D}. Combining these ideas with a dual logarithmic barrier, \cite{mohammadisiahroudi2025quantum} describes a QIPM framework that yields a speedup over classical IPMs in terms of the problem dimension.

While quantizing IPMs in this manner has led to polynomial speedups in the problem dimension $n$, the current approach has two major drawbacks: \textit{(i)} existing QIPMs assume a strong input model (i.e., the quantum RAM model) that is not always justified in a practical setting, and speedups are heavily dependent on this model; \textit{(ii)} the iterative uses of QLSAs rely on an expensive tomography procedure. While it is possible that there is a gap between the algorithm's theoretical worst-case running time and its practical performance, the latter cannot yet be studied: the data structures and computational primitives utilized by existing QIPMs are beyond the capabilities of near-term quantum devices. A detailed resource analysis by Dalzell et al.~\cite{dalzell2022end} for the specific application of portfolio optimization found that a QIPM based on the combined use of QLSA and tomography would require a $T$-gate count of $10^{28}$ when $n = 100$. Informally, this implies that it would take a quantum computer millions of years to solve a problem that could be solved in seconds on a personal laptop. It is important to note that \cite{dalzell2022end} analyzes a relatively old version of QIPMs, and recent results have improved the worst-case theoretical dependence of QIPMs on multiple problem parameters, so Dalzell et al.'s resource analysis may not apply to state-of-the-art QIPMs.

A day before the first version of this paper appeared on the arXiv, Apers and Gribling \cite{apers2023quantum} gave a QIPM that does not make use of QLSAs, thereby avoiding dependence on condition number. This framework leads to asymptotic speedups over classical IPMs for ``tall'' LO problems (i.e., for LO problems with may more constraints than variables), in which the number of constraints is much larger than the number of variables, i.e., $m \gg n$. That said, there is no speedup over classical IPMs unless $m = \Omega (n^{10})$, and the authors assume that each entry of the problem data of $(A,b,c)$ is of the order $\Ocal (\polylog (m,n))$: we do not have this assumption. Like earlier QIPMs, their framework relies on QRAM\footnote{QRAM denotes quantum-accessible storage. The QRAM model of computation is more powerful than the standard model of computation for quantum computers, and there are operations that may be exponentially faster in the QRAM model compared to the standard model.} and block-encodings, and due to results from Clader et al.~\cite{clader2022quantum}, the resources needed to block-encode classical data are a key contributor to the resource estimates of QIPMs \cite{dalzell2022end}.

More generally, the use of quantum subroutines in QIPMs only serves as an attempt to accelerate a single step of the classical IPM. An inherent drawback of this approach is that QIPMs produce a sequence of iterates that exactly mimics the classical IPM trajectory, and therefore lack the potential to exploit quantum effects or improve on the worst-case iteration complexity bound. This raises the question as to whether there exists a ``more naturally quantum'' algorithm that solves optimization problems by tracking the central path.   

In the recent work \cite{leng2023quantum}, Leng, Hickman, Li and Wu proposed a truly quantum analogue to classical (accelerated) gradient descent named \textit{Quantum Hamiltonian Descent} (QHD). Their work can be viewed as a quantization of the Bregman-Lagrangian framework for analyzing the continuous-time dynamics of gradient descent introduced by Wibisono, Wilson, and Jordan~\cite{wibisono2016variational}. The QHD framework admits an implementation on current quantum devices; this is a (practical) advantage over existing QIPMs, for which the underlying quantum primitives are out of reach for current quantum computers. The guiding principle of QHD is to first recast the task of optimization as a dynamical system, and then quantize the resulting continuous-time dynamics. This novel approach offers a new paradigm for the design of quantum algorithms for optimization, which previously, amounted to using quantum subroutines to accelerate a single step of a classical optimization algorithm. Interestingly, the study of dynamical systems arising from IPMs was an active area of research following the introduction of Karmarkar's algorithm in 1984. Many works \cite{anstreicher1988linear, bayer1989nonlinearI, bayer1989nonlinearII, faybusovich1991hamiltonian, faybusovich1995hamiltonian, megiddo1989pathways} studied the trajectories generated by the so-called \textit{Newton barrier flow}. Megiddo \cite{megiddo1989pathways} and Bayer and Lagarias \cite{bayer1989nonlinearI, bayer1989nonlinearII} provided (classical) Lagrangian and Hamiltonian dynamical systems for path-following methods that use a logarithmic barrier function. Faybusovich~\cite{faybusovich1995hamiltonian} later generalized these ideas to convex optimization. Yet, the design and complexity analysis of an algorithm based on these ideas was not considered in these works: although these papers precisely characterize the central path, these efforts did not translate into new algorithms.

\subsection{Contributions}
We use a previously unexplored connection between the central path and the Schr\"odinger equation to develop a new quantum algorithm for LO, that we call the \textit{Quantum Central Path Method} (QCPM). This is achieved by proposing a family of Hamiltonian operators over the positive orthant (a Hamiltonian is an operator characterizing the total energy of a system), parameterized along a one-dimensional axis, say, time. The Hamiltonians encode the behavior of the central path. Specifically, we show that one can approximately follow the central path by simulating a Schr\"odinger equation associated with this time-dependent Hamiltonian. The ground state of each Hamiltonian (i.e., the eigenvector with the lowest eigenvalue) encodes a squeezed Gaussian-like distribution centered at an $\varepsilon$-optimal solution to the primal dual LO pair \eqref{e:LP}-\eqref{e:LP-D}. The Hamiltonian $\Hcal$ that we consider is a Schr\"odinger operator that has two terms, representing the kinetic and potential energy, respectively. In analogy with classical mechanics, this is the standard way to characterize the total energy of a system of one or more particles, whose evolution is then described by the Schr\"odinger equation. Our main result is informally stated as follows:
\begin{theorem}[Main result, informal]
\label{t:informal}
Let $\Hcal(\mu(t), h(t))$ be the time-dependent quantum Hamiltonian defined in \eqref{eqn:Hcal_mu}. Then, for every value of $\mu(t)$, the ground state of $\Hcal(\mu(t), h(t))$ encodes a probability distribution centered on the central path of \eqref{e:LP}-\eqref{e:LP-D}. For $\varepsilon > 0$, we describe in Section~\ref{ss:QCPM} a quantum algorithm that returns classical vectors $x \in \R{n}$ and $y \in \R{m}$ satisfying:
        \begin{align*} 
        A_{i,\cdot} x &\geq b_i \quad \forall i \in [m], \quad x \geq0, \nonumber \\
        \left( A^{\top} \right)_{i,\cdot} y &\leq c_i \quad \forall i \in [n],  \quad  ~y \geq 0, 
        \end{align*}
and a condition akin to
$$ c^{\top} x - b^{\top} y \leq \varepsilon.$$
for the self-dual embedding model. The algorithm can be implemented using 
$$\widetilde{\Ocal}
\left(\frac{(m+n)^{3/2}}{\varepsilon} + \frac{(m+n)^{1/2}}{\varepsilon^{3/2}}\right)
$$
queries to an evaluation oracle for $\Phi$. If the problem data $(A,b,c)$ is stored in binary, one can implement an evaluation oracle for $\Phi$ in the standard gate model using $\Ocal (\nnz(A))$ elementary gates, assuming that $\nnz(A) \ge m+n$.
\end{theorem}
A more precise statement is given in Theorem~\ref{t:QIPMComplexity}. The distinction in Theorem~\ref{t:informal} regarding the optimality guarantee $c^{\top} x - b^{\top} y \leq \varepsilon$ is an important detail: we do not solve the original problem directly, but rather the self-dual embedding model \cite{ye1994nl,roos2005interior}, which we discuss in Section~\ref{ss:SDE}. The optimality gap that we guarantee is with respect to such model, and it cannot be directly compared to that for the original model. This is a common approach in theory and in practice for classical interior point methods \cite{o2016conic,roos2005interior}. We remark that in practice, classically one can employ the self-dual embedding model with a termination criterion that tests for sufficient convergence in the original model; we believe that this approach is viable for our quantum algorithm as well, although we did not study it in detail.

The QCPM inherits many desirable properties from IPM theory (such as provable convergence to strictly complementary solutions), while only computing zeroth-order information on the potential of the Hamiltonian. We emphasize that the QCPM, unlike (Q)IPMs, is not an iterative algorithm: the desired solution can be obtained in \emph{one shot} by measuring the final quantum state. This could potentially provide a significant advantage over the state of the art (Q)IPMs, which require $\Ocal \left(  \min \left\{ \sqrt{m}, \sqrt{n} \right\}  \log \left(\frac{1}{\varepsilon} \right) \right)$ iterations to reach an $\varepsilon$-optimal solution.  Constructing the ground state of a time-dependent Schr\"odinger equation does not come for free, and this determines the complexity of our algorithm. To traverse the ground-state path of the central-path Hamiltonian we use a result by Boixo, Knill, and Somma~\cite{boixo2010fast}, tailored to our specific setup: this provides a fast quantum algorithm that has near-optimal dependence for the eigenpath traversal task, i.e., the task of constructing the ground states of a sequence of closely-related Hamiltonians. 

The idea for the QCPM is rooted in the early works on the Newton barrier flow discussed above, with some significant differences: we use the \emph{self-dual embedding} model \cite{ye1994nl} to obtain an explicit starting point on the central path. We also work \emph{directly} with the nonlinear complementarity equations in the perturbed KKT conditions that define the central path, rather than use a local linearization as is done in IPMs. To the best of our knowledge, the QCPM is the first algorithm that takes this perspective, providing a novel way to quantize classical IPMs compared to existing QIPMs \cite{apers2023quantum,mohammadisiahroudi2025quantum}.

To contextualize the result stated in Theorem~\ref{t:informal}, we provide a comparison of our algorithm to the current state of the art algorithms for solving LO problems in both the classical and quantum models of computation in Table~\ref{tab:compare}. Our algorithm is \textit{not} reliant on access to a classical-write/quantum-read RAM (QRAM): it only requires access to the natural binary description for the LO problem data $(A,b,c)$. Like the QIPM from Apers and Gribling~\cite{apers2023quantum}, our algorithm avoids a condition number dependence. We also emphasize that the QIPMs found in~\cite{apers2023quantum,mohammadisiahroudi2025quantum} assume access to QRAM. In the table we compare our QCPM to classical IPMs, the Primal-Dual Hybrid Gradient (PDHG) algorithm, and several quantum algorithms.

The PDHG algorithm with restarts \cite{applegate2023faster} can also solve linear optimization problems. The PDHG algorithm achieves polylogarithmic dependence on precision using restarts. The number of restarts depends on the Hoffman constant $\kappa_H$ of the KKT system associated to \eqref{e:LP}-\eqref{e:LP-D}. This constant is notoriously difficult to compute and rigorously bound \cite{pena2024easily}; for general problems, it is often left ``as is'' and not expressed as a function of the input size. Our algorithm's performance matches that of PDHG when $\kappa_H = \widetilde{\Ocal}\left((m+n)^{3/2}\,\varepsilon^{-1}+(m+n)^{1/2}\,\varepsilon^{-3/2}\right)$, but we are not aware of explicit bounds for $\kappa_H$ in the general case so a direct comparison with our algorithm is difficult.

Quantum algorithms for zero-sum games \cite{van2019games, bouland2023quantum, li2023regret, gao2024logarithmic} achieve sublinear running times in $m$ and $n$, and polynomial dependence on precision and an $\ell_1$-norm upper bound $R_1$ on the size of the solution to \eqref{e:LP}-\eqref{e:LP-D}. There are notable differences in the query and input models used by both algorithms. The queries made by the QCPM are evaluation queries for a potential function, whereas algorithms for zero-sum games make queries to block-encodings of the problem data. As for the input model, the algorithms found in \cite{van2019games, bouland2023quantum, li2023regret, gao2024logarithmic} rely on QRAM, while the QCPM does not. Algorithms for zero-sum games also employ different definitions of feasibility and optimality than (Q)IPMs and our QCPM, therefore a direct comparison is difficult;  we do not discuss this class of algorithms further.

\begin{table}[ht!]
    \centering
    \caption{Complexity to solve the primal-dual pair \eqref{e:LP}-\eqref{e:LP-D} to precision $\varepsilon$. The $\widetilde{\Ocal}$ notation does not report polylogarithmic terms in all input parameters. The queries mentioned in the ``query complexity'' column access different types of data, and are not necessarily equivalent. }\label{tab:compare}
    \resizebox{\textwidth}{!}{
    \begin{tabular}{llllc}
    \toprule \hline 
    \textbf{Classical Algorithms} &  \textbf{Time complexity} &  &    \\ \hline 
    IPM &  $\widetilde{\Ocal} (\sqrt{n} \left( \nnz(A) + n^2) \right)$ & \cite{lee2015efficient} \\
    IPM  & $\widetilde{\Ocal} ((m + n)^{\omega} )$ & \cite{cohen2021solving, van2020deterministic}  \\ 
    PDHG &  $\widetilde{\Ocal} ( \kappa_H \cdot \nnz (A) )$ & \cite{applegate2023faster} ($\kappa_H$ is the Hoffman constant)  \\
    \toprule \hline 
    \textbf{Quantum Algorithms} &  \textbf{Query complexity} &     \textbf{Gate complexity} & & \textbf{QRAM}\\ \hline 
    QMMWU  & $\widetilde{\Ocal} \left(\sqrt{m + n} R_1^{2.5} \varepsilon^{-2.5 }+ R_1^3 \varepsilon^{-3}  \right)$  &$\widetilde{\Ocal} \left(\sqrt{m + n} R_1^{2.5} \varepsilon^{-2.5}+ R_1^3 \varepsilon^{-3} \right)$ & \cite{bouland2023quantum, gao2024logarithmic} & \cmark\\
    QIPM & $\widetilde{\Ocal} \left( \sqrt{m} n^5   \right)$ &   $\widetilde{\Ocal} \left(\sqrt{m n} \left(n^{6.5} s^2 + n^{\omega + 2} \right)   \right)$ & \cite{apers2023quantum} & \cmark\\ 
    QIPM & $\widetilde{\Ocal} \left( \kappa^2 n^{3/2}   \right)$  & $\widetilde{\Ocal} \left( \kappa^2 n^{3/2}   \right) + \widetilde{\Ocal} \left(n^{2}   \right) \text{ classical (} \kappa \text{ is a condition number)}$ & \cite{mohammadisiahroudi2025quantum} & \cmark\\
    QCPM (this work) &$\widetilde{\Ocal} \left((m+n)^{3/2}\,\varepsilon^{-1}+(m+n)^{1/2}\,\varepsilon^{-3/2} \right)$ &$\widetilde{\Ocal} \left( [(m+n)^{3/2}\,\varepsilon^{-1}+(m+n)^{1/2}\,\varepsilon^{-3/2}]\,\nnz(A)\right)$    & Thm~\ref{t:QIPMComplexity} & \xmark \\
    \bottomrule
    \end{tabular}}
\end{table}

The QCPM proposed in this paper provides (for the first time, as far as we are aware) convincing evidence that an end-to-end quantum speedup for LO might be possible even in the circuit model without QRAM. While this would require further enhancements, in particular in terms of the error dependence, we view our framework as promising because it has no known direct classical analogue. Although many of the technical ingredients for our algorithm and its analysis were present in the literature in some form, we believe that the QCPM idea could provide a new, potentially interesting avenue for constrained convex optimization on quantum computers. This paper lays the foundations for a rigorous analysis of algorithms based on the idea of simulating the central path using a quantum computer. It may be useful to perform a detailed resource analysis of the QCPM framework, in the same spirit as the analysis of Dalzell et al.~\cite{dalzell2022end} for QIPMs; we leave this for future work. It is important to remark that we did not attempt to construct a classical analogue of the quantum dynamical system proposed here. Doing so might lead to new classical algorithms, but we expect that it would take significant ingenuity to design a classical system that remains in the global minimum of the time-varying potential, if at all possible.

Next, we provide a more detailed summary of the technical overview of the paper's main results. 
 
\subsection{Technical summary and roadmap}
Hamiltonian simulation is a fundamental task in quantum computation as it separates the computational power between quantum and classical~\cite{osborne2012hamiltonian}. For a given time-dependent Hamiltonian operator $\Hcal(t)$ and an initial state $\ket{\psi_0}$, the Hamiltonian simulation task is to simulate the dynamics governed by the \textit{Schr\"odinger equation}:
$$
    \bfi \frac{\td }{\td t} \ket{\psi (t)} = \Hcal (t) \ket{\psi(t)}, \text{ subject to }\ket{\psi(0)} = \ket{\psi_0},
$$
where $\bfi$ denotes the imaginary unit (for readers unfamiliar with the braket notation used in quantum mechanics: $\ket{\cdot}$ denotes a column vector, $\bra{\cdot}$ denotes its conjugate transpose). In this paper, we propose a Quantum Central Path Method (QCPM) relying on the efficient simulation of Schr\"odinger dynamics on quantum computers. Specifically, we design a quantum Hamiltonian (known as the central-path Hamiltonian) that \emph{encodes} the behavior of the central path in its ground state. By preparing an initial state representing a point exactly on the central path and simulating the ground-state trajectory from $t=0$ to some $t=T$, with the property that at $t=T$ the ground state is centered at an (approximate) optimal solution to the primal-dual LO pair \eqref{e:LP}-\eqref{e:LP-D}, we are able to recover an approximate solution by simply measuring the final state. 

Our QCPM shares some features with Quantum Hamiltonian Descent (QHD)~\cite{leng2023quantum} in that \textit{(i)} both can be regarded as \emph{quantizations} of classical dynamical processes (path-following in QCPM, accelerated gradient descent in QHD), \textit{(ii)} both are formulated as quantum dynamics governed by Schr\"odinger equations, and \textit{(iii)} in both cases the ground state of the Hamiltonian operator encodes the solution to an optimization problem. The central-path Hamiltonian can be interpreted as a \emph{global} linearization of the central path: this is essentially different from any existing QIPMs, where each Newton step represents a \emph{local} linearization of the central path. Consequently, our algorithm can be more efficient because there is no need for iterative uses of QLSA and state tomography.

The QCPM traces the ground-state trajectory using the eigenpath traversal algorithm of Boixo, Knill, and Somma~\cite{boixo2010fast}; see Section~\ref{ss:traversalalg}. This requires copies of the ground state of $\Hcal$ for the initial time parameter, which in our case corresponds to the complementarity deviation $\mu$ of IPMs. The self-dual embedding model of Ye, Todd, and Mizuno~\cite{ye1994nl} places the primal-dual pair~\eqref{e:LP}--\eqref{e:LP-D} in a problem of dimension $m+n+2$ with a known central-path point. A review of the self-dual embedding and the initial point on the central path is provided in Section~\ref{ss:SDE}.

Under the self-dual embedding framework, the central path is (uniquely) characterized by element-wise complementarity and positivity constraints on the variable $z$ and its associated slack $s(z)$: 
$$\Ccal \coloneqq \left\{ z(\mu) : \mu \in (0,1] \right\} = \left\{ z \in \R{m + n + 2} :  z_i \cdot s(z)_i = \mu,~z_i , s(z)_i > 0,~ \forall i \in \{1, \dots,m + n + 2 \}\right\}.$$
Classical IPMs track the central path by iteratively applying Newton's method, which requires locally linearizing the nonlinear complementarity equation 
$$z_i \cdot s(z)_i = \mu \quad \forall i \in \{ 1, \dots, m + n + 2 \},$$
and successively decreasing $\mu$ by a constant factor after each Newton step. We propose a different approach: a $\mu$-dependent Schr\"odinger operator on the positive orthant that 
directly incorporates the nonlinear complementarity equation as a potential field,
$$\Hcal(\mu, h) \coloneqq - \frac{h^2}{2} \nabla^2 + \Phi_{\mu}(z),$$
where the potential function is essentially the log-barrier function $-\sum_i \log(z_i s(z)_i/\mu)$ clipped at the endpoint of a suitable domain. By construction,
the global minimum of the potential function $\Phi_\mu(z)$ corresponds to the $\mu$-center on the central path.
In the eigenpath traversal, we need a discrete-time, monotonically decreasing schedule for the values of $h$ and $\mu$ to ensure an efficient transition between successive ground states.
The central-path Hamiltonian $\Hcal(\mu(t), h(t))$ is inspired by the canonical time-independent Hamiltonian in quantum mechanics:
$$
    \Hcal \coloneqq - \nabla^2 + V,
$$
where $\nabla^2$ denotes the Laplacian operator in $\R{n}$, accounting for kinetic energy, and $V\colon \R{n} \mapsto \R{}$ is a potential function, accounting for potential energy. In our Hamiltonian $\Hcal(\mu, h)$, the role of the parameter $h$ is to control the variance of the probability distribution encoded in the ground state of the central-path Hamiltonian. In the usual setting, a (small) fixed value $\hbar$ known as \textit{Planck's constant} is used, but this choice is insufficient to guarantee convergence in our setting: by decreasing $h$ over the course of the ground-state traversal, we ensure that \textit{(i)} we concentrate on a probability mass centered at the optimal solution; and \textit{(ii)} classical solutions obtained by sampling from this distribution are contained in a small neighborhood of the optimum. The detailed construction of the central-path Hamiltonian is provided in Section \ref{s:QIPM}. Properly designing the schedule for $h$ and $\mu$
is a crucial technical insight of our approach, and more details are provided in Section~\ref{ss:Qsim}.

Equipped with the Hamiltonian $\Hcal(\mu, h)$ defined above, the remaining challenge is two-fold: we need to \textit{(i)} verify that simulating the Schr\"odinger equation associated to our central-path Hamiltonian does in fact solve our optimization problem; and \textit{(ii)} rigorously analyze the cost of doing so. We establish the correctness of our approach through a combination of results already found in the literature and new results of our own. In Section \ref{ss:HamiltonianDesign}, we show that the potential function $\Phi_\mu(z)$ is strongly convex and its Hessian is lower bounded by a uniform constant that only depends on the optimal solution (see Lemma~\ref{lem:uniform-strong-convexity}).
This strongly convex potential has two immediate implications: \textit{(i)} the spectral gap of the central-path Hamiltonian $\Hcal(\mu, h)$ has a lower bound that scales linearly in the parameter $h$, and \textit{(ii)} the ground state of the central-path Hamiltonian $\Hcal(\mu, h)$ is concentrated in a neighborhood of the central path.
In Section~\ref{ss:concentration} we show that with sufficiently small $h$, the ground state of $\Hcal(\mu, h)$ encodes a distribution centered on the central path and with small variance.
In Section \ref{ss:Qsim} we discuss the construction of a discrete-time schedule of $h$ and $\mu$ for the simulation. Section~\ref{ss:groundstate} analyzes the preparation of the ground state of the initial Hamiltonian. In Section~\ref{ss:discretization-block-encoding} we analyze the costs for all the above operations. Section~\ref{ss:QCPM} combines these results to demonstrate that properly choosing a discrete-time schedule of $h$ and $\mu$ allows one to solve \eqref{e:LP}-\eqref{e:LP-D} via traversing the ground-state trajectory of our central-path Hamiltonian: the final output is a classical description of an $\varepsilon$-optimal solution, and no intermediate measurements or refinements are required. This is in contrast with existing quantum algorithms for LO \cite{van2019games, augustino2021quantum, bouland2023quantum, kerenidis2020quantum, mohammadisiahroudi2023inexact}, which require either $\Ocal \left(\sqrt{n} \log \frac{1}{\varepsilon} \right)$ or $\Ocal \left(\frac{\log (n)}{\varepsilon^2} \right)$ calls to a Hamiltonian simulation subroutine (to either solve quantum linear systems or prepare Gibbs states). 

To achieve our complexity result, we combine the requirements for our convergence guarantees with the cost of spatial discretization (see Section~\ref{ss:discretization-block-encoding}) and implementing the eigenpath traversal algorithm~\cite{boixo2010fast}. We prove that the QCPM returns a classical description of an $\varepsilon$-precise solution to \eqref{e:LP}-\eqref{e:LP-D} using at most 

$$\widetilde{\mathcal{O}}\Big(\left((m+n)^{3/2}\,\varepsilon^{-1} + (m+n)^{1/2}\,\varepsilon^{-3/2}\right)\cdot \nnz(A)\Big)$$

elementary gates. The overall complexity suggests that in the low-precision regime, the size of a potential quantum speedup is determined jointly by problem sparsity, solution size, and the $\varepsilon^{-3/2}$ simulation dependence introduced by the central-path schedule. We stress that the speedups achieved here are independent of QRAM. We also believe our techniques should readily generalize to other classes of constrained convex optimization problems. 

\paragraph{Relation to the conference version.}
In an earlier version of this work, presented at QIP 2025, we used a different potential function, and our analysis of the gap was based on a harmonic oscillator approximation of $\Hcal(\mu, h)$ that holds in the limit $h \to 0$. By relying on the strong convexity of the potential $\Phi_\mu(z)$, we are able to give a lower bound on the spectral gap of the operator $\Hcal(\mu, h)$ without the assumption that $h \to 0$. The algorithm to construct the ground state for the final Hamiltonian has also been changed with respect to the paper presented at QIP 2025.

\paragraph{Organization.}
In Section \ref{s:Prelim} we define some notation and review the theory of classical interior point methods applied to solving LO problems, before providing results on the adiabatic theorem for unbounded Hamiltonians and algorithms for simulating Schr\"odinger equations. Section \ref{s:QIPM} concerns the design and analysis of the Quantum Central Path Method, and compares our overall complexity result to the current state-of-the-art. Section~\ref{s:numerical} gives a brief numerical demonstration of QCPM. Section \ref{s:con} concludes the paper.

\section{Preliminaries}\label{s:Prelim}
We distinguish the quantity $a$ to the $k$-th power and the value of $a$ at iteration $k$ using round brackets, writing $a^k$ and $a^{(k)}$ to denote these quantities, respectively. We write $[n]$ to represent the set $\{1, \dots, n\}$. 

We denote the $i$-th element of a vector $x \in \R{n}$ by $x_i$ for $i \in [n]$, and the $ij$-th element of a matrix $A \in \R{m \times n}$ by $A_{ij}$ for $i \in [m]$ and $j \in [n]$. To refer to the $i$-th row of a matrix $A$, we write $A_{i, \cdot}$ and write $A_{ \cdot, j}$ when referring to its $j$-th column. For a vector $x \in \R{n}$, the matrix $\diag (x) \in \R{n \times n}$ takes the values of $x$ along its diagonal and zero elsewhere. For vectors $u, v \in \R{n}$, $u \odot v$ denotes their \textit{Hadamard} (or, entry-wise) product:
$$ u \odot v = \begin{bmatrix}
    u_1 \cdot v_1 \\
    \vdots \\
    u_n \cdot v_n
\end{bmatrix}.$$
We write $0_{n} \in \R{n}$ when referring to an all-zeros vector of length $n$, and $0_{m \times n} \in \R{m \times n}$ denotes the $m \times n$ all-zeros matrix.

The smallest and largest singular values of a matrix $A$ are denoted $\sigma_{\min}(A), \sigma_{\max}(A)$, and the smallest and largest eigenvalues are denoted $\lambda_{\min}(A), \lambda_{\max}(A)$. The operator norm of $A$ is defined as $\| A \| \coloneqq \sigma_{\max} (A)$. As is common in quantum mechanics, we often write \emph{ground state} to refer to the normalized eigenvector with the smallest eigenvalue of a Hamiltonian, and \emph{ground energy} to refer to the corresponding eigenvalue.

We let $\Scal_+^n$ and $\Scal_{++}^n$ represent the spaces of symmetric positive semidefinite, and symmetric positive definite matrices in $\R{n \times n}$, respectively. For $U, V \in \Scal^n$, we write $U \succeq V$ ($U \succ V$) to indicate that the matrix $U - V$ is symmetric positive semidefinite (symmetric positive definite), i.e., $U - V \in \Scal^n_+$ ($U - V \in \Scal^n_{++}$).

\subsection*{Order estimates}
We define $\Ocal (\cdot)$ as
$$f(x) = \Ocal(g(x)) \iff \exists \ell \in \R{}, \alpha \in \R{}_+,~\text{such that}~f(x) \leq \alpha g(x)\quad \forall x > \ell.$$
We write $f(x) = \Omega (g(x)) \iff g(x) = \Ocal(f(x))$. If there exists positive constants $\alpha_1$ and $\alpha_2$ such that  
$$\alpha_1 g(x) \leq f(x) \leq \alpha_2 g(x) \quad \forall x > 0,$$
then we write $f(x) = \Theta (g(x))$.

We also define $\widetilde{\Ocal} (f(x)) =\Ocal(f(x) \cdot \textup{polylog}(f(x)))$ and when the function depends polylogarithmically on other variables we write 
$$\widetilde{\Ocal}_{\alpha, \beta}~(f(x))=\Ocal(f(x)\cdot\textup{polylog}(\alpha, \beta, f(x))).$$

\subsection{Interior Point Methods for Linear Optimization}\label{sec:ipm-basics}
In this section we outline the classical IPM theory applied to solving linear optimization problems.  

\subsubsection{Primal and dual LO problems}
Recall that we are interested in the primal LO problem and its dual 
$$
        \min_{x \in \Pcal}\quad c^{\top} x,\quad \quad \max_{y \in \Dcal}\quad b^{\top} y,
$$ 
where $\Pcal$ and $\Dcal$ are the \textit{primal and dual feasible sets}, defined as 
$$
    \Pcal \coloneqq \left\{ x \in \R{n} : Ax \geq b,~x \geq 0 \right\}, \quad
    \Dcal \coloneqq \left\{ y \in \R{m} : A^{\top} y \leq c,~y \geq 0 \right\}.
$$
Likewise, the sets of \textit{interior feasible solutions} of \eqref{e:LP} and \eqref{e:LP-D} are given by 
$$
    \interior \left( \Pcal \right) \coloneqq \left\{ x \in \R{n} : Ax > b,~x > 0 \right\}, \quad
    \interior \left( \Dcal \right) \coloneqq \left\{ y \in \R{m} : A^{\top} y < c,~y > 0 \right\}.
$$
The well known property of \textit{weak duality} always holds, and asserts that any primal feasible $x \in \Pcal$ provides an upper bound $c^{\top} x$ on the value $b^{\top} y$, and conversely, any dual feasible $y \in \Dcal$ provides a lower bound $b^{\top} y$ on the value $c^{\top} x$. The nonnegative quantity
$$ c^{\top} x - b^{\top} y$$
is referred to as the \textit{duality gap} associated to the pair $(x,y) \in \Pcal \times \Dcal$. 

Whenever $(x_*, y_*) \in \Pcal \times \Dcal$ exhibit vanishing duality gap, i.e., $c^{\top} x_* - b^{\top} y_* = 0$, then $x_*$ is an optimal solution to \eqref{e:LP}, and $y_*$ is an optimal solution to \eqref{e:LP-D}. Hence, we define the optimal set of \eqref{e:LP}-\eqref{e:LP-D} to be
$$\Pcal \Dcal_* \coloneqq \left\{ (x,y) \in \Pcal \times \Dcal :  c^{\top} x = b^{\top} y \right\}. $$
Clearly, $\Pcal \Dcal_*$ is nonempty if and only if the inequality system 
\begin{equation}\label{e:OptConditions}
    \begin{aligned}
        Ax &\geq b, \quad &x \geq 0, \\
        -A^{\top} y &\geq -c, \quad &y \geq 0, \\
         b^{\top} y - c^{\top} x &\geq 0,
    \end{aligned}
\end{equation}
is solvable. Introducing a \textit{homogenizing variable} $\beta$, system \eqref{e:OptConditions} is equivalent to the homogeneous system 
\begin{equation}\label{e:homogenousEmbedding}
    \begin{bmatrix}
    0_{m \times m} & A & - b \\
    -A^{\top} & 0_{n \times n} & c \\
    b^{\top} & - c^{\top} & 0 
    \end{bmatrix}
    \begin{bmatrix}
    y \\ x \\ \beta 
    \end{bmatrix} \geq 
    \begin{bmatrix}
    0_m \\ 0_n \\ 0 
    \end{bmatrix},  \quad x \geq 0,~y\geq 0,~\beta \geq 0,
\end{equation}
when $\beta = 1$. Indeed, one can verify that any $(x,y,\beta)$ which solves \eqref{e:homogenousEmbedding} with $\beta > 0$ gives rise to a solution $\left(\frac{x}{\beta},\frac{y}{\beta}, 1 \right)$ to \eqref{e:OptConditions}. More concisely, letting
$$ \Mbar \coloneqq \begin{bmatrix}
    0_{m \times m} & A & - b \\
    -A^{\top} & 0_{n \times n} & c \\
    b^{\top} & - c^{\top} & 0 
    \end{bmatrix}, \quad \zbar\coloneqq \begin{bmatrix}
    y \\ x \\ \beta 
    \end{bmatrix},$$
solving \eqref{e:LP} and \eqref{e:LP-D} to optimality is equivalent \cite[see, Theorem I.3]{roos2005interior} to obtaining a feasible solution to the system 
\begin{equation}\label{e:fease}
    \Mbar \zbar \geq 0, \quad \zbar \geq 0, \quad \beta > 0.
\end{equation}
Next, we discuss the condition under which \eqref{e:fease} is solvable, which is intimately related to the practical concern of starting from an interior feasible solution. 

\subsubsection{The self-dual embedding model}\label{ss:SDE}
We begin with a formal definition of the \textit{interior point condition} (IPC) from \cite{roos2005interior}.
\begin{definition}[Definition I.4 in \cite{roos2005interior}] We say that any system of (linear) equalities and (linear) inequalities satisfies the interior-point condition (IPC) if there exists a solution that satisfies all inequality constraints in the system. 
\end{definition}
The IPC ensures that the primal-dual pair \eqref{e:LP} and \eqref{e:LP-D} has an optimal solution $(x_*, y_*) \in \Pcal \Dcal_*$ whenever there exists a strictly feasible solution $(x,y) \in \interior (\Pcal) \times \interior (\Dcal)$. Being able to easily determine an interior point $(x,y) \in \interior (\Pcal) \times \interior (\Dcal)$ is also a practical concern: (feasible) IPMs are initialized to a strictly feasible starting point, and na\"ively determining a strictly feasible solution to \eqref{e:LP} and \eqref{e:LP-D} is as challenging as solving these problems to optimality. Here we discuss a way to embed the problems \eqref{e:LP} and \eqref{e:LP-D} into a slightly larger one that always has a trivial strictly feasible solution, and can be readily solved with an IPM, following a strategy first introduced by Ye, Todd and Mizuno~\cite{ye1994nl}. 

While system \eqref{e:fease} does not satisfy the IPC, this can be reconciled upon introducing another auxiliary variable $\vartheta \geq 0$, 
and adding one additional row and column to $\Mbar$ as follows: 
\begin{align}\label{eqn:M_and_r}
    M \coloneqq \begin{bmatrix}
    \Mbar & r \\ 
    -r^{\top} & 0 
\end{bmatrix} =  \begin{bmatrix}
    \begin{bmatrix}
    0_{m \times m} & A & - b \\
    -A^{\top} & 0_{n \times n} & c \\
    b^{\top} & - c^{\top} & 0 
    \end{bmatrix} & r \\ 
    -r^{\top} & 0 
\end{bmatrix}, \quad z \coloneqq \begin{bmatrix}
    \zbar \\ \vartheta 
\end{bmatrix} = \begin{bmatrix}
    y \\ x \\ \beta \\ \vartheta 
    \end{bmatrix},
\end{align}
where 
$$ r = \bar{e} - \Mbar \bar{e},$$
and $\bar{e}$ is the all-ones vector of length $m+n+1$. Note that $M$ is a \textit{skew-symmetric} matrix of dimension $\nbar \coloneqq n+m+2$, i.e., $M^{\top} = -M$. Defining $q \in \R{\nbar}$ to be the vector 
$$ q = \begin{bmatrix}
    0_{\nbar - 1} \\ \nbar 
\end{bmatrix},$$
one can also see that the system 
\begin{equation}\label{e:SDE1}
    Mz \geq - q, \quad z \geq 0,
\end{equation}
admits the all-ones vector of length $\nbar$ as a strictly feasible solution, which we denote by $e$. If $z$ solves \eqref{e:SDE1} with $\vartheta = 0$, i.e., we have $z = (\zbar, 0)$, then $\zbar$ must be a solution to \eqref{e:fease}. Hence, solving \eqref{e:LP} and \eqref{e:LP-D} can be reduced to finding a solution to \eqref{e:SDE1} with $\vartheta = 0$ and $\beta > 0$.

The \textit{self-dual embedding} of \eqref{e:LP} and \eqref{e:LP-D} is defined as 
\begin{equation}\label{e:self-dual} 
    \min \left\{q^{\top} z: Mz \ge -q,~ z \ge 0\right\}. 
\end{equation}
This problem is called ``self-dual'' because its dual problem 
\begin{equation*} 
    \max \left\{- q^{\top} u: M^{\top} u \leq q,~ u \ge 0\right\}, 
\end{equation*}
can be recognized as equivalent to \eqref{e:self-dual}, upon recalling that $M$ is a skew-symmetric matrix with $M^{\top} = - M$. Like system \eqref{e:SDE1}, the self-dual embedding problem \eqref{e:self-dual} also admits the strictly feasible solution $e$, and thus trivially satisfies the IPC. 

Moving forward, we will write $\Scal \Pcal$ when referring to the set of feasible solutions to \eqref{e:self-dual}. We denote the sets of interior feasible, and optimal solutions by $\interior (\Scal \Pcal)$ and $\Scal \Pcal_*$, respectively. 

The next result from \cite{roos2005interior} asserts that we obtain optimal solutions to \eqref{e:LP} and \eqref{e:LP-D} by solving \eqref{e:self-dual}.
\begin{theorem}[Theorem I.6 in \cite{roos2005interior}]
The system \eqref{e:SDE1} has a solution with $\vartheta = 0$ and $\beta > 0$ if and only if the problem \eqref{e:self-dual} has an optimal solution with $\beta = z_{\nbar-1} > 0$.
\end{theorem}

\subsubsection{The central path and its neighborhood}
For any vector $z\in \R{\nbar}$, we define its slack vector\footnote{As we have defined it here, $s$ really is a \textit{surplus} variable. } 
$$s(z)\coloneqq Mz + q.$$ 
It follows that 
\begin{equation*}
z~\textit{is a feasible solution to \eqref{e:self-dual}}~\iff~ z \geq 0~\textit{and}~s(z) \geq 0.    
\end{equation*}
Recalling that $e$ denotes the all-one vector of length $\nbar$, note that $s(e) = e$ and $e \odot s(e) = e$. Hence, $\mu = 1$ for the point $\left( z, s(z) \right) = (e,e) \in \interior (\Scal \Pcal)$. 

For every positive $\mu$ there exists a \textit{unique} nonnegative vector $z$ such that 
\begin{align}\label{e:quadratic}\tag{CP}
    z \odot s(z) = \mu e,\quad z\ge 0,~ s(z) \ge 0, 
\end{align}
see, \cite[Lemma I.13]{roos2005interior}. We denote the unique nonnegative solution to the quadratic system \eqref{e:quadratic} as $z (\mu)$, which we refer to as the $\mu$-\textit{center}. Using this notation, we may write $z(1) = e$. The set of $\mu$-centers $\Ccal = \left\{ z (\mu): \mu > 0 \right\}$ is the \textit{central path} of the problem \eqref{e:self-dual}, and an optimal solution of \eqref{e:self-dual} is the limiting point of the central path as $\mu \to 0^+$. The central path constitutes an analytic curve $\xi (\mu): (0, \infty) \mapsto \R{\nbar}$, and its graph $(\mu, \xi(\mu))$ satisfies
\begin{align}\label{e:CPgraph}
    z\odot s(z) = \mu e,\quad z > 0,~ s(z) > 0.
\end{align}
Because IPC is satisfied, the central path is guaranteed to exist and is uniquely determined from the starting point $\left( z(1), s(z(1)) \right) = (e,e)$. Put another way, $(e,e)$ is the point \textit{on} the central path corresponding to $\mu = 1.$

The standard IPM can be interpreted as an algorithm that approximately follows the central path by locally linearizing the quadratic system \eqref{e:quadratic}. The solution of each locally linearized system is the so-called \textit{Newton step}. To see this, suppose that $z$ is a positive solution to \eqref{e:self-dual} such that its slack vector $s(z)$ is also positive, i.e., $z$ and $s(z)$ satisfy \eqref{e:CPgraph}. To find the displacement $\Delta z$ such that $z^+ \coloneqq z + \Delta z$ is the $\mu$-center, we want to solve the equation:
\begin{align*}
    (z + \Delta z) \odot s(z + \Delta z) = \mu e.
\end{align*}
Defining $s = s(z)$ and $\Delta s = M \Delta z$, we obtain a nonlinear equation with quadratic term $\Delta z \odot \Delta s$, namely,
\begin{align}\label{e:nonlinear}
    Zs + Z\Delta s + S \Delta z + \diag \left(\Delta z \right) \Delta s = \mu e,
\end{align}
where $Z \coloneqq \diag(z)$ and $S \coloneqq \diag(s)$. If we further assume $z^+$ is in a small neighborhood of $z$, i.e., $\Delta z$ is small (and hence, so is $\Delta s = M \Delta z$), the quadratic term in \eqref{e:nonlinear} can be omitted, and the nonlinear equation reduces to a linear system in the unknowns $\Delta z$ and $\Delta s$. 

The foundation for IPM theory is that performing the local linearization we just described at each iterate, and solving the resulting equation system for $\Delta z$ and $\Delta s$ allows us to make sufficient progress towards the optimal solution (and remain in $\interior (\Scal \Pcal)$), provided that the current iterate $(z, s(z))$ is in some sense close to the central path. To quantify a notion of closeness, one defines a proximity measure that gives rise to a \textit{neighborhood} of the central path. In this paper, we consider the distance metric 
$$ d_2 (z, s; \mu) \coloneqq \left\| z \odot s (z) - \mu e \right\|_2,$$
which for fixed $\gamma \in (0,1)$ gives rise to a narrow neighborhood of the central path:
$$ \Ncal_2 (\gamma) \coloneqq \left\{ (z, s(z)) \in \interior (\Scal \Pcal) : d_2 (z, s; \mu) \leq \gamma \mu \right\}.$$
One can observe that for every $\gamma \in (0,1)$, the following set of inclusions hold: 
$$ \Ccal \subset \Ncal_2 (\gamma) \subset \interior (\Scal \Pcal).$$

\subsection{Quantum eigenpath traversal algorithms}
\label{ss:traversalalg}

For the sake of self-containment, we give a high-level overview of the simulation algorithm in \cite{boixo2010fast} that we use for the Schr\"odinger equation with Hamiltonian $\Hcal(\mu, h) \coloneqq - \frac{h^2}{2} \nabla^2 + \Phi_{\mu}(z)$. Note that the simulation is not adiabatic per se, in the sense that we do not simulate the entire dynamics of the system; rather, we ``jump'' between discrete values of the time parameter, staying in the ground state for those specific values, but losing access to the behavior of the system for the values of $\mu$ between the jumps. This is one way to get around standard complexity results for the adiabatic theorem.

The algorithm of \cite[Thm.~V.II]{boixo2010fast} only requires access to a phase estimation subroutine and knowledge of a $\mu$ schedule that has large overlaps between successive ground states; in Theorem~\ref{t:QIPMComplexity} we prove that such a schedule is available for $\Hcal(\mu, h)$. Let the points on the schedule be $(\mu_j, h_j)$, $j=0,\dots,L$. Given the ground state $\ket{\psi_0(\mu_j, h_j)}$, we can estimate the eigenvalue of the ground state $\ket{\psi_0(\mu_{j+1}, h_{j+1})}$ using phase estimation applied to the Hamiltonian $\Hcal(\mu_{j+1}, h_{j+1})$, because the overlap $|\braket{\psi_0(\mu_j, h_j)}{\psi_0(\mu_{j+1}, h_{j+1})}|$ is large by assumption ($> 1/2$). The probability of success of phase estimation is directly proportional to the overlap, and using a few (i.e., polylogarithmic) number of copies of $\ket{\psi_0(\mu_{j+1}, h_{j+1})}$ in parallel, we can ensure that the estimation is successful with very high probability \cite[Lem.~IV.7]{boixo2010fast} by looking at the phase estimated for the majority of the copies. Once we have the eigenvalue for the next ground state $\ket{\psi_0(\mu_{j+1}, h_{j+1})}$, phase estimation of $\Hcal(\mu_{j+1}, h_{j+1})$ followed by measurement of the eigenvalue register allows us to determine if we have landed onto $\ket{\psi_0(\mu_{j+1}, h_{j+1})}$ or not by checking if the eigenvalue is within a certain range of the target value, where the range depends on the spectral gap. At this point, an alternating measurement protocol that applies projective measurements ${\cal A} = \{\ket{\psi_0(\mu_j, h_j)}, \ket{\psi^\perp_0(\mu_j, h_j)}\}$, ${\cal B} = \{\ket{\psi_0(\mu_{j+1}, h_{j+1})}, \ket{\psi^\perp_0(\mu_{j+1}, h_{j+1})}\}$ suffices to prepare $\ket{\psi_0(\mu_{j+1}, h_{j+1})}$ with high probability: if ${\cal B}$ fails to prepare $\ket{\psi_0(\mu_{j+1}, h_{j+1})}$, we have landed on $\ket{\psi^\perp_0(\mu_{j+1}, h_{j+1})}$; then ${\cal A}$ followed by ${\cal B}$ either prepares $\ket{\psi_0(\mu_{j+1}, h_{j+1})}$, or returns to $\ket{\psi^\perp_0(\mu_{j+1}, h_{j+1})}$. We can keep trying until we succeed, and because of the overlap guarantee, we need a polylogarithmic number of attempts for a target probability of success. \cite{boixo2010fast} does not use projective measurements, but rather reflections around the ground state: these are easy to implement from the projective measurement and an ancilla qubit for a phase flip, using a standard construction; the concept of the algorithm is the same as outlined above. Thus, phase estimation and knowledge of an appropriate schedule $(\mu_j, h_j)$, $j=0,\dots,L$, with large overlap suffice to prepare the final ground state $\ket{\psi_0(\mu_{L}, h_{L})}$ with high probability.

Let $\{\Hcal_j\}_{j=0}^L$, $L\ge1$, be a sequence of Hamiltonians on a common finite-dimensional Hilbert space (these are the same as $\Hcal(\mu_j, h_j)$ discussed earlier), each with a unique normalized ground state $\ket{\psi_j}$, ground energy $E_j$, and spectral gap lower bound $\Delta_j>0$. Choose a positive uniform bound $\Gamma\ge\max_{0\le j\le L}\|\Hcal_j\|$ and define
\begin{align}\label{eq:bks-evolution-oracles}
    \tau:=\frac{\pi}{2\Gamma},\qquad
    U_j:=e^{-i\tau \Hcal_j},\qquad
    \varphi_j:=-\tau E_j,\qquad
    \vartheta_j:=\tau\Delta_j.
\end{align}
The unitary $U_j$ has the same eigenvectors as $\Hcal_j$, with eigenphases in $[-\pi/2,\pi/2]$. Hence, $\vartheta_j$ is a lower bound on the eigenphase gap of $U_j$, separating $\varphi_j$ from its other eigenphases.

Write $\mathrm{PE}(U_j,\xi)$ for the coherent phase-estimation oracle with precision $\xi$, as in \cite{boixo2010fast}. Given an estimate $\tilde\varphi_j$ satisfying $|\tilde\varphi_j-\varphi_j|\le\vartheta_j/5$, the reflection oracle of \cite[Definition~III.3]{boixo2010fast} acts as  $R(U_j,\tilde\varphi_j,\vartheta_j) =R_j:=I-2\ket{\psi_j}\!\bra{\psi_j}$.
We can implement the reflection oracle by phase detection, a conditional sign flip, and uncomputation, see \cite{boixo2010fast}. We assume access to controlled $U_j$ and $U_j^\dagger$, and use phase estimation and the reflection oracles in the following statement.

\begin{theorem}[Theorem~V.2 in \cite{boixo2010fast}]
\label{thm:bks-eigenpath-traversal}
Suppose the ground states $\ket{\psi_j}$ of $\{\Hcal_j\}_{j=0}^L$ satisfy
\begin{align*}
    |\langle\psi_{j-1}|\psi_j\rangle|^2
    >\frac43\left(\frac12+\beta\right),
    \qquad j=1,\ldots,L,
\end{align*}
for some $\beta\in(0,1/4)$. Given $r\ge1$ copies of $\ket{\psi_0}$, there is a procedure that produces $\ket{\psi_L}^{\otimes r}$ together with estimates
$|\tilde\varphi_j-\varphi_j|\le\vartheta_j/5$, $j=1,\ldots,L$, up to joint error amplitude at most $Le^{-r\beta^2}$.
The procedure uses at most $2Lr+1$ calls to the phase-estimation oracles $\mathrm{PE}(U_j,\vartheta_j/5)$ and their inverses, and a random number $M_R$ of reflection-oracle calls satisfying
\begin{align*}
    \mathbb E[M_R]<2Lr,\qquad
    \mathbb E[\zeta^{M_R}]\le\zeta^{3Lr},
    \qquad 1\le\zeta\le\frac74.
\end{align*}
\end{theorem}
\begin{proof}
We obtain an initial estimate $|\tilde\varphi_0-\varphi_0|\le\vartheta_0/5$ with one preliminary call to $\mathrm{PE}(U_0,\vartheta_0/5)$ on $\ket{\psi_0}$. Then we apply the one-step transformation of \cite[Lemma~IV.8]{boixo2010fast} successively to $U_{j-1}$ and $U_j$, using the target phase gap $\vartheta_j$ and each reflection's own gap and phase estimate. Errors accumulate linearly over the $L$ transitions, and the number of oracle calls follows by \cite[Theorem~V.2]{boixo2010fast}.
\end{proof}

Note that, in addition to the costs for the oracles in Theorem~\ref{thm:bks-eigenpath-traversal}, one should include the cost for the initial state preparation.

\section{A quantum algorithm that traces the central path}\label{s:QIPM}
In this section we propose a Hamiltonian formalism for the central path of linear optimization problems. We show how one can solve LO problems by simulating the Schr\"odinger equation associated with our Hamiltonian, and provide a rigorous complexity analysis of the resulting scheme. Fix a target central-path parameter $\varepsilon\in(0,1)$. Throughout this section, we consider $\mu\in[\varepsilon,1]$ and construct an algorithm that prepares a ground state concentrated near the central-path point $z(\varepsilon)$. The target $\varepsilon$ remains fixed as $\mu$ varies.

\subsection{Quantum representation of the central path}\label{ss:HamiltonianDesign}

In this subsection, we construct a family of quantum Hamiltonian operators $\Hcal(\mu, h)$, where $\mu \in [\varepsilon,1]$ is the central-path parameter. We show that the ground state of $\Hcal(\mu, h)$ gives rise to a distribution of points concentrated near the $\mu$-center $z(\mu)$. Therefore, by measuring the ground state of $\Hcal(\mu, h)$ we can approximate the $\mu$-center.

For every $\mu \in [\varepsilon, 1]$, we assume the $\mu$-center $z(\mu)$ and the slack vector $s(\mu) := Mz(\mu)+q$ satisfy
\begin{align*}
    z_i(\mu)s_i(\mu)=\mu,\qquad 0<z_i(\mu),s_i(\mu)\le R_\infty,\qquad \text{for~all}\quad i=1,\ldots,\nbar.
\end{align*}
Set $\tau_l := \varepsilon/(2R_{\infty})$ and $\tau_u := 2R_{\infty}$. We define the quadratically continued logarithm, denoted by $\phi(x)$, as a function that agrees with $-\log(x)$ on $[\tau_l, \tau_u]$ and is quadratic outside this interval:
\begin{align}\label{eq:clipped-log-definition}
\phi(x) :=
\begin{cases}
    -\log(\tau_l)-\dfrac{x-\tau_l}{\tau_l}
        +\dfrac{(x-\tau_l)^2}{2\tau^2_l}, & x\le\tau_l,\\[1.2ex]
    -\log(x), & \tau_l\le x\le \tau_u,\\[1.2ex]
    -\log(\tau_u)-\dfrac{x-\tau_u}{\tau_u}
        +\dfrac{(x-\tau_u)^2}{2\tau_u^2}, & x\ge \tau_u.
\end{cases}
\end{align}
This definition helps to ensure strong convexity of the potential function (discussed below), while avoiding the issue of unbounded derivatives. Then, we define the following potential function, denoted by $\Phi_\mu(z)$:
\begin{align}\label{eq:regularized-potential}
    \Phi_\mu(z) := \frac{q^\top z}{\mu} +\sum_{i=1}^{\nbar}\phi(z_i) +\sum_{i=1}^{\nbar}\phi(s_i(z)) -c_\mu, \qquad \text{where}\quad c_\mu:=\nbar(1-\log\mu).
\end{align}
Using the skew-symmetry of $M$ and the complementarity condition $z(\mu)\odot s(\mu) = \mu e$, it is straightforward to verify that the potential $\Phi_\mu(z)$ vanishes at $z =z(\mu)$. Furthermore, $\Phi_\mu(z)$ is nonnegative and strictly convex over $\R{\nbar}$. Therefore, $z = z(\mu)$ is the unique global minimizer of $\Phi_\mu(z)$.
\begin{lemma}\label{lem:uniform-strong-convexity}
    The potential function $\Phi_\mu(z)$ is twice continuously differentiable over $\R{\nbar}$. Moreover, for any $\mu \in [\varepsilon, 1]$, the $\mu$-center $z(\mu)$ is the unique global minimizer of $\Phi_\mu(z)$ and for every $z \in \R{\nbar}$,
    \begin{align}\label{eq:global-Hessian-bound}
        \nabla^2\Phi_\mu(z) \succeq\frac1{4R_{\infty}^2}I,\qquad \Phi_\mu(z) \ge \frac{1}{8R^2_\infty}\|z - z(\mu)\|^2.
    \end{align}
\end{lemma}
\begin{proof}
By definition, $\phi(x) \in C^2(\R{})$ and $\frac{1}{\tau^2_u} \le \phi''(x) \le \frac{1}{\tau^2_l}$ for every $x \in \R{}$.
We have
\begin{align}\label{eqn:Phi_hessian_representation}
\nabla^2\Phi_\mu(z) =D_z(z)+M^\top D_s(z)M
\succeq\frac1{\tau_u^2}(I+M^\top M)
\succeq\frac1{\tau_u^2}I,
\end{align}
where
$D_z(z):=\diag\bigl(\phi''(z_1),\ldots,\phi''(z_{\nbar})\bigr)$, $D_s(z):=\diag\bigl(\phi''(s_1(z)),\ldots,\phi''(s_{\nbar}(z))\bigr)$.
Plugging in $\tau_u = 2R_\infty$, and recalling that $\Phi_\mu(z(\mu)) = 0$, $\nabla \Phi_\mu(z(\mu)) = 0$, proves~\eqref{eq:global-Hessian-bound}.
\end{proof}

We now introduce the quantum Hamiltonian that encodes the central path $z(\mu)$.

\begin{definition}[Central-path Hamiltonian]
    For $\varepsilon \le \mu \le 1$ and $h > 0$, we define a two-parameter family of elliptic operators over $\R{\nbar}$,
    \begin{align}\label{eqn:Hcal_mu}
        \Hcal(\mu, h) = - \frac{h^2}{2}\nabla^2 + \Phi_\mu(z),
    \end{align}
    where $\nabla^2 = \sum^{\nbar}_{j=1}\frac{\partial^2}{\partial z_j^2}$ is the Laplacian operator, and $\Phi_\mu(z)$ is defined in~\eqref{eq:regularized-potential}.
\end{definition}

For fixed $\mu$ and $h$, the operator $\Hcal(\mu, h)$ is a so-called Schr\"odinger operator that acts on square-integrable (i.e., $L^2(\R{\bar{n}})$) functions. 
We note that the Hilbert space $L^2(\R{\bar{n}})$ is infinite-dimensional (i.e., it cannot be expressed as the span of finitely many basis vectors). Consequently, the Hamiltonian $\Hcal(\mu, h)$ should be interpreted as an infinite-dimensional operator and thus does not admit a finite-dimensional matrix representation. Because a gate-based quantum computer can only perform computations in a finite-dimensional Hilbert space, appropriate discretization schemes of $\Hcal(\mu, h)$ must be considered in the actual quantum algorithm. The discretization schemes enable us to pick a specific basis set and reduce the infinite-dimensional operator $\Hcal(\mu, h)$ into a finite-dimensional matrix, which is then represented and computed using a quantum computer. This is a standard consideration in the quantum simulation of systems with infinitely many degrees of freedom (e.g., unbounded operators, partial differential equations, etc.) with extensive theoretical and numerical treatment~\cite{liu2021efficient,childs2022quantumSim,an2021time}, see Section~\ref{ss:discretization-block-encoding}.

\begin{theorem}[Spectral gap estimate]
\label{thm:spectral-gap-estimate}
    For every $\mu \in [\varepsilon, 1]$ and $h > 0$, the operator $\Hcal(\mu, h)$ has a purely discrete spectrum and
    $
        \mathrm{gap}(\Hcal(\mu, h)) \ge \frac{h}{2R_\infty},
    $
    where $\mathrm{gap}(\cdot)$ denotes the spectral gap (i.e., the difference between the lowest and the second-lowest eigenvalue) of the operator. 
\end{theorem}
\begin{proof}
    The fact that the spectrum is purely discrete follows from the standard Schr\"odinger operator theory; see~\cite{hislop2012introduction}. For a Schr\"odinger operator defined over a bounded smooth convex domain with a strongly convex potential $\Phi$ such that $\nabla^2 \Phi \ge m I$, Yau~\cite[Theorems 1.1, 1.2]{yau2009estimate} proves that the spectral gap has a lower bound $h\sqrt{m}$. Using a standard extension argument, we generalize this lower bound over the full Euclidean space $\R{\nbar}$, see Appendix~\ref{ss:convex-shrodinger-gap} for details. 
\end{proof}

\subsection{Concentration of quantum central path}
\label{ss:concentration}

The main goal of this section is to provide a concentration guarantee for the ground state of the central-path Hamiltonian $\Hcal(\mu, h)$. Specifically, we give a condition on the parameter $h$ that enforces membership in the $\Ncal_2(\gamma)$-neighborhood at the final stage $\mu=\varepsilon$, when we need to ensure that a measurement of the quantum state yields points distributed in a small neighborhood of the solution to the LO problem.

In the following lemma we characterize the lower envelope of the potential function $\Phi_\mu(z)$ in terms of the distance metric $d_2(z;\mu):= \|z \odot s(z) - \mu e\|_2$, which defines the $\Ncal_2$ neighborhood of the central path.

\begin{lemma}\label{lem:potential_lower_profile}
    For every $\mu \in [\varepsilon, 1]$ and $0 < \gamma < 1$, let $\mathcal{G}_{\gamma, \mu} := \{z \in \R{\nbar} \colon \tau_l \le z_i, s_i \le \tau_u,~d_2(z,\mu) \le \gamma \mu\}$. Then for all $z \notin \mathcal{G}_{\gamma, \mu}$, we have $\Phi_\mu(z) \ge \frac{\gamma^2}{12}$.
\end{lemma}
\begin{proof}
    For $a \in [\tau_l, \tau_u]$, define the Bregman divergence of the function $\phi$ as
        $D_\phi(x,a) := \phi(x) - \phi(a) - \phi'(a)(x-a)$.
    First we claim that for every $z \in \R{\nbar}$, the potential $\Phi_\mu(z)$ has the representation
    \begin{align}\label{eq:Bregman-representation}
    \Phi_\mu(z) =\sum_{i=1}^{\nbar}D_\phi(z_i,z_i(\mu))
     +\sum_{i=1}^{\nbar}D_\phi(s_i(z),s_i(\mu)).
    \end{align}
    To show this, note that $\phi'(z_i(\mu))= s_i(\mu)/\mu$, $\phi'(s_i(\mu))= z_i(\mu)/\mu$. Plugging these into the right-hand side of~\eqref{eq:Bregman-representation}, we get
    \begin{align*}
    \mathrm{RHS}&=\sum_{i=1}^{\nbar}\left[\phi(z_i) + \phi(s_i) - \phi(z_i(\mu)) - \phi(s_i(\mu))\right] - \frac{1}{\mu}\sum_{i=1}^{\nbar}\left[s(\mu)^\top(z-z(\mu)) + z(\mu)^\top s(\mu)\right]\\
        &= \sum_{i=1}^{\nbar}\left[\phi(z_i) + \phi(s_i) - \phi(z_i(\mu)) - \phi(s_i(\mu))\right] - \frac{1}{\mu} \left(s(\mu) - Mz(\mu)\right)^\top (z - z(\mu))\\
        &= \sum_{i=1}^{\nbar}\left[\phi(z_i) + \phi(s_i) - \phi(z_i(\mu)) - \phi(s_i(\mu))\right] + \frac{q^\top z}{\mu} - \frac{q^\top z(\mu)}{\mu}\\
        &= \Phi_\mu(z) - \Phi_\mu(z(\mu)) = \Phi_\mu(z),
    \end{align*}
    where in the second and third steps we used $s(\mu) = Mz(\mu)+q$ and the last step follows from $\Phi_\mu(z(\mu)) = 0$.

    For $z \notin \mathcal{G}_{\gamma,\mu}$, there are two possible cases: (i) $z_i, s_i \in [\tau_l, \tau_u]$ but $d_2(z;\mu) \ge \gamma \mu$, and (ii) some $z_i$ or $s_i$  does not belong to $[\tau_l, \tau_u]$. Next, we prove that $\Phi_\mu(z) \ge \frac{\gamma^2}{12}$ holds in both cases. 
    
    \begin{enumerate}
        \item[(i)] First, for $z_i, s_i \in [\tau_l, \tau_u]$ we can readily check that (define $\rho(r) := r-1-\log(r)$)
        \begin{align*}
            \Phi_\mu(z) = \sum^{\nbar}_{i=1} \rho\left(\frac{z_i s_i(z)}{\mu}\right).
        \end{align*}
        Set $r_i :=z_is_i(z)/\mu>0$. If every $r_i\in[1/2,3/2]$, then $\rho''(r)=r^{-2}\ge4/9$ on each line segment between $1$ and $r_i$. Because $\rho(1)=\rho'(1)=0$, strong convexity gives $\sum_{i=1}^{\nbar}\rho(r_i) \ge\frac29\|r-e\|_2^2$.
        Note that we also have $d_2(z;\mu)=\mu\|r-e\|_2$, so when $d_2(z;\mu) \ge \gamma \mu$, we have 
        \begin{align*}
            \Phi_\mu(z) = \sum_{i=1}^{\nbar}\rho(r_i) \ge \frac{2}{9}\gamma^2 > \frac{1}{12}\gamma^2.
        \end{align*}
        If there is some $r_i$ that lies outside $[1/2,3/2]$, monotonicity of $\rho$ on either side of $1$ gives
        \begin{align*}
            \rho(r_i) \ge\min\{\rho(1/2),\rho(3/2)\} >\frac1{12} \ge\frac{\gamma^2}{12}.
        \end{align*}

    \item[(ii)] Now we discuss the case when some $z_i$ or $s_i(z)$ does not belong to $[\tau_l,\tau_u]$. For a fixed $a \in [2\tau_l, \tau_u/2]$, the function $D_\phi(x,a)$ is non-increasing for $x \le \tau_l$ and is non-decreasing for $x \ge \tau_u$. Therefore, if $x \le \tau_l$, then
    $
        D_\phi(x,a) \ge D_\phi(\tau_l,a) = \log\left(\frac{a}{\tau_l}\right) + \frac{\tau_l}{a} - 1 \ge \log(2) - \frac{1}{2} \ge \frac{1}{12} \ge \frac{\gamma^2}{12}.
    $
    And similarly, if $x \ge \tau_u$, then
    $
        D_\phi(x,a) \ge D_\phi(\tau_u,a) = \log\left(\frac{a}{\tau_u}\right) + \frac{\tau_u}{a} - 1 \ge 1 - \log(2) \ge \frac{1}{12} \ge \frac{\gamma^2}{12}.
    $
    \end{enumerate}
    Therefore, in both cases, we have $\Phi_\mu(z) \ge \gamma^2/12$, which completes the proof.
\end{proof}

Next, we bound the ground energy of the central-path Hamiltonian using a standard Gibbs-type test function. This result plays a crucial role in our main concentration theorem in this subsection.

\begin{lemma}[Ground-energy estimate]\label{lem:ground-energy-upper-bound-estimate}
    Suppose the LO data satisfy $\|A\| \le 1$, $\|b\|\le 1$, and $\|c\|\le 1$.
    For $\varepsilon \le \mu \le 1$, let $E_0(h)$ be the ground energy (i.e., the lowest eigenvalue) of the central-path Hamiltonian~\eqref{eqn:Hcal_mu} with parameter $h$. Then $E_0(h) \le \frac{9 h \nbar R_\infty}{\varepsilon}$.
\end{lemma}
\begin{proof}
    In the following argument, $\mu$ is fixed and $h>0$ varies.
    For $\beta>0$, define a test function $\varphi_\beta(z):=Z_\beta^{-1/2}e^{-\beta\Phi_\mu(z)/2}$ with $Z_\beta:=\int_{\Rmbb^{\nbar}}e^{-\beta\Phi_\mu(z)}\,dz$. 
    Note that, by virtue of the normalization factor $Z_\beta^{-1/2}$, $\|\varphi_\beta\|=1$.
    Because $E_0(h) = \min_{\|\psi\|=1} \langle \psi|\Hcal(\mu, h)|\psi\rangle$, it follows that 
    \begin{align}\label{eqn:rayleigh-quotient}
        E_0(h) &\le \langle \varphi_\beta|\Hcal(\mu, h)|\varphi_\beta\rangle = \frac{h^2}{2}\langle\varphi_\beta |(-\nabla^2)|\varphi_\beta \rangle + \langle\varphi_\beta |\Phi_\mu(z)|\varphi_\beta \rangle\notag\\
        &= \underbrace{\frac{h^2\beta^2}{8Z_\beta} \int_{\R{\nbar}} \|\nabla \Phi_\mu\|^2 e^{-\beta \Phi_\mu(z)}\,d z}_{\text{kinetic~term}}  + \underbrace{\frac{1}{Z_\beta} \int_{\R{\nbar}} \Phi_\mu(z) e^{-\beta \Phi_\mu(z)}\,d z}_{\text{potential~term}}.
    \end{align}
    For the kinetic term, because the function $\partial_i \Phi_\mu(z)e^{-\beta\Phi_\mu(z)}$ vanishes at $z = \pm \infty$, the integration-by-parts formula gives $0 = \sum^{\nbar}_{i=1} \int_{\R{\nbar}} \partial_i \left(\partial_i \Phi_\mu(z)e^{-\beta\Phi_\mu(z)}\right)$,
    which implies that \begin{align}\label{eqn:ground-energy-kinetic}
        \frac{1}{Z_\beta}\int_{\R{\nbar}} \|\nabla \Phi_\mu\|^2 e^{-\beta \Phi_\mu(z)}\,d z = \frac{1}{\beta} \sum^{\nbar}_{i=1} \int_{\R{\nbar}} \partial^2_i \Phi_\mu(z) \frac{e^{-\beta \Phi_\mu(z)}}{Z_\beta}\,dz \le \frac{1}{\beta} \cdot \sup_{z\in \R{\nbar}}\Tr[\nabla^2 \Phi_\mu(z)].
    \end{align}
    
For the potential term, we note that $\Phi_\mu(z)$ is convex and $\Phi_\mu(z(\mu))=0$ is a global minimum for every $\mu\in [\varepsilon,1]$. This implies that 
\begin{align}\label{eq:convex-gradient-upper}
    \Phi_\mu(z) = \Phi_\mu(z) - \Phi_\mu(z(\mu)) \le \langle \nabla \Phi_\mu(z), z - z(\mu)\rangle.
\end{align}
Because $(z_i-z_i(\mu))e^{-\beta \Phi_\mu(z)}$ vanishes at $z = \pm \infty$ for all $i\in [\nbar]$, integration by parts yields
\begin{align*}
    0 = \int_{\R{\nbar}} \partial_i \left[(z_i-z_i(\mu))e^{-\beta \Phi_\mu(z)}\right]\,dz \implies Z_\beta = \beta  \int_{\R{\nbar}}(z_i - z_i(\mu))\partial_i \Phi_\mu(z) e^{-\beta \Phi_\mu(z)}\,dz.
\end{align*}
Summing over $i \in [\nbar]$ and using~\eqref{eq:convex-gradient-upper} gives 
\begin{align}\label{eqn:ground-energy-potential}
    \frac{1}{Z_\beta} \int_{\R{\nbar}} \Phi_\mu(z) e^{-\beta \Phi_\mu(z)}\,dz \le \frac{\nbar}{\beta}.
\end{align}
Plugging~\eqref{eqn:ground-energy-kinetic} and~\eqref{eqn:ground-energy-potential} into~\eqref{eqn:rayleigh-quotient}, we have 
\begin{align}\label{eqn:E0_with_beta}
    E_0(h) \le \frac{h^2\beta}{8} \cdot \sup_{z\in \R{\nbar}}\Tr[\nabla^2 \Phi_\mu(z)] + \frac{\nbar}{\beta}.
\end{align}
    Using $\phi''(x) \le 4R_\infty/\varepsilon$, \eqref{eqn:Phi_hessian_representation}, and $\|M\|^2_F \le 33\nbar$ (see~Lemma~\ref{lem:M-Frobenius-bound}) we have
    \begin{align*}
        \Tr[\nabla^2 \Phi_\mu(z)] = \Tr[D_z(z)] + \Tr[M^\top D_s(z) M] \le \frac{4R^2_\infty}{\varepsilon^2} \left(\nbar + \|M\|^2_F\right) \le \frac{136 \nbar R^2_\infty}{\varepsilon^2}.
    \end{align*}
    Using this bound in~\eqref{eqn:E0_with_beta} and minimizing the right-hand side over $\beta > 0$ we obtain the desired upper bound, concluding the proof.
\end{proof}

\begin{lemma}\label{lem:M-Frobenius-bound}
Suppose the LO data satisfy $\|A\|\le1$, $\|b\|_2\le1$, and $\|c\|_2\le1$.  Let $\nbar$, $\overline{M}$, and $M$ be the same as in Section~\ref{sec:ipm-basics}. Then$\|M\|_{\mathrm F}^2\le33 \nbar$.
\end{lemma}
\begin{proof}
    The block structure and the input normalization give $\|\overline{M}\|_{\mathrm F}^2 =2\|A\|_{\mathrm F}^2+2\|b\|_2^2+2\|c\|_2^2 \le2\operatorname{rank}(A)\|A\|^2+4 \le \nbar+2$.
    Note that $\|r\| \le 4\sqrt{\nbar -1}$. It follows that $\|M\|_{\mathrm F}^2=\|\widetilde M\|_{\mathrm F}^2+2\|r\|_2^2 \le \nbar +2+32(\nbar-1)\le33 \nbar$.
\end{proof}

\begin{theorem}[Ground-state concentration]
    For every $\mu \in [\varepsilon, 1]$ and $0 < \gamma < 1$, let $\Psi_0(\mu,h)$ denote the ground state of the central-path Hamiltonian $\Hcal(\mu, h)$. Suppose that the LO data is normalized: $\|A\|,\|b\|, \|c\|\le 1$. Then,
    \begin{align*}
        \mathbb{P}_{z\sim|\Psi_0(\mu,h)|^2}\left[d_2(z;\mu) \ge \gamma \mu\right] \le \frac{108h\nbar R_\infty}{\varepsilon\gamma^2}.
    \end{align*}
\end{theorem}
\begin{proof}
    Because the kinetic energy is nonnegative, using Lemma~\ref{lem:potential_lower_profile}, we have
    \begin{align*}
        E_0(h) \ge \mathbb{E}_{z\sim|\Psi_0(\mu,h)|^2}[\Phi_\mu] \ge \frac{\gamma^2}{12} \mathbb{P}_{z\sim|\Psi_0(\mu,h)|^2}[z \notin \mathcal{G}_{\gamma, \mu}].
    \end{align*}
    Because $\{d_2(z;\mu) \ge \gamma \mu\}\subset \{z \notin \mathcal{G}_{\gamma,\mu}\}$, it follows from Lemma~\ref{lem:ground-energy-upper-bound-estimate} that 
    \begin{align*}
        \mathbb{P}_{z\sim|\Psi_0(\mu,h)|^2}[d_2(z;\mu) \ge \gamma \mu] \le \frac{12 E_0(h)}{\gamma^2} \le \frac{108 h \nbar R_\infty}{\varepsilon\gamma^2}. 
    \end{align*}
    This concludes the proof.
\end{proof}

\begin{corollary}
\label{cor:final-concentration-scale}
Suppose that the LO data is normalized: $\|A\|,\|b\|, \|c\|\le 1$. 
Choosing $h = h_\star :=\frac{\gamma^2\varepsilon}{324\nbar R_\infty}$ gives a per-run success probability of at least $2/3$.
\end{corollary}

\subsection{Discrete traversal of the quantum central path}\label{ss:Qsim}

We have shown that the $\mu$-center $z(\mu)$ is encoded in the ground state of the Hamiltonian operator $\Hcal(\mu, h)$. We can therefore solve LO problems by preparing the ground state of $\Hcal(\mu, h)$ for sufficiently small $\mu$. To prepare such ground state, in this subsection we construct a sequence of Hamiltonians whose successive ground states have constant overlap.

To obtain the ground-state concentration guaranteed by Corollary~\ref{cor:final-concentration-scale}, we use a two-stage traversal schedule.
The first stage transports the ``center'' of the ground state along the central path from $z(1)=e$ to $z(\varepsilon)$ by decreasing $\mu$ from $1$ to $\varepsilon$ while setting $h=\mu$. This stage ends at $h_f:=\varepsilon$.
The second stage keeps $\mu=\varepsilon$ fixed and narrows the wave packet around $z(\varepsilon)$ by decreasing $h$ to the value $h_\star$ in Corollary~\ref{cor:final-concentration-scale}.
For fixed $\gamma\in(0,1)$ and $R_\infty\ge1$, we have $h_\star / h_f = \gamma^2/(324\nbar R_\infty) = \Ocal(1/\nbar)$. Thus, for sufficiently large $\nbar$, the first stage alone does not reach the value of $h$ required for the concentration bound of Corollary~\ref{cor:final-concentration-scale}, and a final squeezing stage is needed. We assume throughout this subsection that $h_\star<h_f$ and include both stages in a single discrete schedule.

\begin{remark}
    We choose this two-stage schedule because it allows us to construct a discrete traversal in which every pair of successive ground states of the central-path Hamiltonian has large overlap. Simpler traversal schedules may exist, but we do not expect them to yield asymptotically-better complexity guarantees.
\end{remark}

Suppose that $\|A\|,\|b\|,\|c\|\le1$, and fix $0<\gamma<1$. Recall that $\nbar=n+m+2$, $\tau_l=\varepsilon/(2R_\infty)$, and $\tau_u=2R_\infty$.
To directly specify the values of $\mu$ that constitute the schedule, define
\begin{align}\label{eq:filter-G}
    G(\mu)
    :=\frac43R_\infty^{3/2}\sqrt{\nbar}
         (\mu^{-3/2}-1)
      +\frac{24R_\infty^2\sqrt{\nbar}}{\varepsilon}\log\frac1\mu,
    \qquad \mu\in[\varepsilon,1].
\end{align}
The function $G$ is motivated by a point-wise estimate on the rate of change of the ground-state projector (see~Lemma~\ref{lem:projector-variation}) and is strictly decreasing, with $G(1)=0$. Thus, its inverse maps $[0,G(\varepsilon)]$ to $[\varepsilon,1]$. Define $R_\star:=\log\frac{h_f}{h_\star}$, $C_\star:=\frac{24R_\infty^2\sqrt{\nbar}}{\varepsilon}$, and the numbers of steps as
\begin{align}\label{eq:two-stage-step-counts}
    J:=\left\lceil\sqrt{\frac32}\,G(\varepsilon)\right\rceil,
    \qquad J_\star:=\left\lceil\sqrt{\frac32}\,C_\star R_\star\right\rceil,
    \qquad L:=J+J_\star.
\end{align}
Both $J$ and $J_\star$ are positive. The two stages are
\begin{align}\label{eq:two-stage-schedule}
    \mu_j:=\begin{cases}
        G^{-1}\!\left(\dfrac{jG(\varepsilon)}J\right), & 0\le j\le J,\\[1.2ex]
        \varepsilon, & J<j\le L,
    \end{cases}
    \qquad
    h_j:=\begin{cases}
        \mu_j, & 0\le j\le J,\\[1.2ex]
        h_f e^{-(j-J)R_\star/J_\star}, & J<j\le L.
    \end{cases}
\end{align}
Thus $(\mu_0,h_0)=(1,1)$, $(\mu_J,h_J)=(\varepsilon,h_f)$, and $(\mu_L,h_L)=(\varepsilon,h_\star)$. Although $\Hcal(\mu,h)$ has two arguments, the prescribed schedule defines a single-parameter family: in the first stage $h=\mu$, while in the second stage $\mu=\varepsilon$ is fixed and $h$ varies. The schedule defines the discrete path
\begin{align*}
    \Hcal_j:=\Hcal(\mu_j,h_j)
      =-\frac{h_j^2}{2}\Delta+\Phi_{\mu_j},
    \qquad j=0,\ldots,L.
\end{align*}
Let $\psi_j$ denote the normalized positive ground state of $\Hcal_j$.

\begin{remark}
    The points $(\mu_j,h_j)$ for $j \le J$ in the schedule~\eqref{eq:two-stage-schedule} can be computed by scalar bisection because $G$ is strictly decreasing, and the remaining points require only elementary arithmetic. 
    We assume these points can be computed up to an arbitrary precision; under standard assumptions on the computational model, a sufficiently accurate finite-precision representation changes the number of steps in the schedule by at most a constant factor.
\end{remark}
\begin{lemma}\label{lem:projector-variation}
Under the assumptions of this subsection, let $P(\mu,h)$ denote the ground-state projector of $\Hcal(\mu,h)$. Along the first stage $\left\|\frac{\d}{\d\mu}P(\mu,\mu)\right\| \le -G'(\mu)$ for $\mu \in [\varepsilon,1]$, and along the second stage $\|\partial_hP(\varepsilon,h)\|\le\frac{C_\star}{h}$ for $h\in[h_\star,h_f]$.

\end{lemma}
\begin{proof}
For the first stage, define $\lambda=1/\mu$, so $1\le\lambda\le1/\varepsilon$ and $h=1/\lambda$. Write
\begin{align*}
    \Hcal(1/\lambda,h)
    =-\frac{h^2}{2}\Delta+\lambda W+V_0-c_\lambda I,
\end{align*}
where $W(z):=q^\top z$, $V_0(z):=\sum_{i=1}^{\nbar}\phi(z_i)+\sum_{i=1}^{\nbar}\phi(s_i(z))$, $c_\lambda:=\nbar(1+\log\lambda)$.
The term $-c_\lambda I$ changes neither the ground state nor the spectral gap. Thus, consider (we define an operator $T(\lambda)$ because it is used later)
\begin{align*}
    \Hcal_0(\lambda):= \Hcal(1/\lambda, 1/\lambda) + c_\lambda I = T(\lambda)+\lambda W+V_0,
\qquad T(\lambda):=-\frac{1}{2\lambda^2}\Delta.
\end{align*}
Let $\psi_0(\lambda)$, $E(\lambda)$, and $P_0(\lambda)$ denote the normalized positive ground state, ground energy, and ground-state projector of $\Hcal_0$, respectively. In particular, $P_0(\lambda)=P(1/\lambda,1/\lambda)$.
The Hessian representation~\eqref{eqn:Phi_hessian_representation}, Lemma~\ref{lem:M-Frobenius-bound}, and Theorem~\ref{thm:spectral-gap-estimate} give
\begin{align}\label{eq:filter-analytic-inputs}
    \nabla^2\Phi_\mu&\succeq K:=\frac{I+M^\top M}{\tau_u^2},
    \qquad \Delta\Phi_\mu\le\frac{136\nbar R_\infty^2}{\varepsilon^2}, \qquad
    \operatorname{gap}(\Hcal_0(\lambda))\ge \frac{1}{\tau_u\lambda}.
\end{align}
The Laplacian bound uses $\phi''\le4R_\infty^2/\varepsilon^2$ and $\|M\|_{\mathrm F}^2\le33\nbar$. Also, $M^\top=-M$ and $q=(I-M)e$ imply
$q^\top(I+M^\top M)^{-1}q=e^\top(I+M)\bigl[(I+M)(I-M)\bigr]^{-1}(I-M)e=\nbar$.

The clipped potential is $C^2$, uniformly convex, and has bounded Hessian. Analytic perturbation theory for the common quadratic-form domain, together with simplicity of the ground eigenvalue, imply that the Hamiltonian has a differentiable positive ground state~\cite[Chapter~VII]{kato1995perturbation}. Define 
\begin{align*}
    Q_0(\lambda):=I-P_0(\lambda),\qquad S(\lambda):=Q_0(\lambda)(\Hcal_0(\lambda)-E(\lambda))^{-1}Q_0(\lambda),
\end{align*}
where the inverse is taken on $\operatorname{Ran}(Q_0(\lambda))$ (this is the range of $Q_0$, i.e., the set of all outputs that $Q_0$ can produce). In what follows, without ambiguity, we omit the $\lambda$ dependence and write $\Hcal_0$, $\psi_0$, $E$, $P_0$, and $Q_0$ for simplicity; primes denote derivatives with respect to $\lambda$. 
Differentiating $\Hcal_0\psi_0=E\psi_0$ and using $\langle\psi_0,\psi_0'\rangle=0$ yields $\psi_0'=-SQ_0\Hcal_0'\psi_0$.
Because $P_0'=\ketbra{\psi_0'}{\psi_0}+\ketbra{\psi_0}{\psi_0'}$ and $\psi_0'\perp\psi_0$, we have
\begin{align}\label{eq:new-projector-response}
    \|P_0'\|=\|\psi_0'\|=\|SQ_0\Hcal_0'\psi_0\|.
\end{align}
Moreover, the norm of $S$ is controlled by $(\Hcal_0-E)^{-1}$ on the range of $Q_0$, which is upper bounded by the spectral gap of $\Hcal_0$: $\|S\|\le1/\mathrm{gap}(\Hcal_0)$ and $S^2\preceq S/\mathrm{gap}(\Hcal_0)$. Now $\Hcal_0'=W-2T/\lambda$, while the ground-state equation gives $Q_0T\psi_0=-Q_0\Phi_{1/\lambda}\psi_0$. Hence,
\begin{align}\label{eqn:P-prime}
    \|P_0'\|
    \le\|SQ_0W\psi_0\|
       +\frac2\lambda\|SQ_0\Phi_{1/\lambda}\psi_0\|.
\end{align}

For the first term, hold $h=1/\lambda$ fixed and apply Lemma~\ref{lem:filter-force} to $-h^2\Delta/2+V_0+\xi^\top z$, with the matrix $K$ in~\eqref{eq:filter-analytic-inputs}. Denote its ground energy and normalized positive ground state by $E_h(\xi)$ and $\psi_{h,\xi}$. The Hellmann--Feynman identity states that
$\nabla_\xi E_h(\xi)=\int_{\R{\nbar}}z|\psi_{h,\xi}(z)|^2\,\d z$.
Differentiate this identity in the direction $q$, using $\psi_0'=-SQ_0\Hcal_0'\psi_0$ for variations of $\xi$ at fixed $h$. At $\xi=\lambda q$, Lemma~\ref{lem:filter-force} then gives
\begin{align*}
    2\langle Q_0W\psi_0,SQ_0W\psi_0\rangle
    =-q^\top\nabla_\xi^2E_h(\lambda q)q
    \le q^\top K^{-1}q=\tau_u^2\nbar.
\end{align*}
Writing $\|SQ_0W\psi_0\|^2$ as an inner product, then using $S^2\preceq S/\mathrm{gap}(\Hcal_0)$, the previous displayed equation, and Theorem~\ref{thm:spectral-gap-estimate} we obtain
\begin{align}\label{eqn:P-prime-contribution-1}
    \|SQ_0W\psi_0\|^2
    \le\frac{\tau_u^2\nbar}{2\mathrm{gap}(\Hcal_0)}
    \le\frac{\tau_u^3\nbar\lambda}{2}.
\end{align}

For the second term, set $\rho=\psi_0^2$. Lemma~\ref{lem:filter-variance} and~\eqref{eq:filter-analytic-inputs} imply
\begin{align*}
    \|Q_0\Phi_{1/\lambda}\psi_0\|^2
    =\operatorname{Var}_{\rho}(\Phi_{1/\lambda})
    \le\frac{h^2}{4}\mathbb E_\rho\Delta\Phi_{1/\lambda}
    \le\frac{h^2}{4}\frac{136\nbar R_\infty^2}{\varepsilon^2}.
\end{align*}
Because $h/\mathrm{gap}(\Hcal_0)=\tau_u$, it follows that
\begin{align}\label{eqn:P-prime-contribution-2}
    \frac2\lambda\|SQ_0\Phi_{1/\lambda}\psi_0\|
    \le\frac{\tau_u\sqrt{136\nbar}R_\infty}{\varepsilon\lambda}.
\end{align}
Substituting~\eqref{eqn:P-prime-contribution-1} and~\eqref{eqn:P-prime-contribution-2} into~\eqref{eqn:P-prime}, and using $\tau_u=2R_\infty$ and $\sqrt{136}\le12$, gives
\[
    \|P_0'(\lambda)\|
    \le 2R_\infty^{3/2}\sqrt{\nbar}\sqrt\lambda
       +\frac{24R_\infty^2\sqrt{\nbar}}{\varepsilon\lambda}.
\]
Because $P(\mu,\mu)=P_0(1/\mu)$, the chain rule multiplies this bound by $\mu^{-2}$ and proves that
$\left\|\frac{\d}{\d\mu}P(\mu,\mu)\right\|
\le 2R_\infty^{3/2}\sqrt{\nbar}\mu^{-5/2}
+\frac{24R_\infty^2\sqrt{\nbar}}{\varepsilon\mu}
=-G'(\mu)$. 

For the second stage of the schedule, keep $\mu=\varepsilon$ fixed and write
\begin{align*}
    r:=\log(h_f/h),\qquad h(r)=h_fe^{-r},\qquad
    \Hcal_\star(r)=T_\star(r)+\Phi_{\varepsilon},\qquad
    T_\star(r)=-\frac{h(r)^2}{2}\Delta.
\end{align*}
Let $\psi_\star(r)$, $P_\star(r)$, $Q_\star(r)$, and $S_\star(r)$ be the normalized positive ground state, its projector, the complementary projector, and the reduced resolvent, respectively. The identity~\eqref{eq:new-projector-response} applies to this family as well, with derivatives taken with respect to $r$. Because $T_\star'=-2T_\star$ and $Q_\star T_\star\psi_\star=-Q_\star\Phi_{\varepsilon}\psi_\star$, we have
\begin{align*}
    \|P_\star'(r)\| =2\|S_\star Q_\star T_\star\psi_\star\| \le\frac{2\tau_u}{h(r)}
    \sqrt{\operatorname{Var}_{\psi_\star^2}(\Phi_{\varepsilon})}
    \le\tau_u\sqrt{\mathbb E_{\psi_\star^2}\Delta\Phi_{\varepsilon}}
     \le\frac{2\sqrt{136}R_\infty^2\sqrt{\nbar}}{\varepsilon}
     \le C_\star.
\end{align*}
The first inequality uses the gap bound $\mathrm{gap}(\Hcal_0) \le h(r)/\tau_u$, the second uses Lemma~\ref{lem:filter-variance}, and the third uses~\eqref{eq:filter-analytic-inputs}. Finally, $P_\star(r)=P(\varepsilon,h(r))$ and $h'(r)=-h(r)$, so $\|\partial_hP(\varepsilon,h)\|=\|P_\star'(r)\|/h$. This concludes the proof.
\end{proof}

\begin{proposition}\label{prop:two-stage-schedule}
For the schedule~\eqref{eq:two-stage-schedule}, every pair of successive ground states satisfies
\begin{align}\label{eq:successive_overlap}
    |\langle\psi_{j-1}|\psi_j\rangle|
    \ge\frac1{\sqrt3},\qquad j=1,\ldots,L.
\end{align}
Each $\Hcal_j$ has spectral gap at least $h_j/\tau_u$. The numbers of first- and second-stage steps satisfy
\begin{align*}
    J=\Ocal\!\left(
       1+\frac{R_\infty^{3/2}\sqrt{\nbar}}{\varepsilon^{3/2}}
       +\frac{R_\infty^2\sqrt{\nbar}}{\varepsilon}\log\frac1{\varepsilon}
       \right),\qquad J_\star=\Ocal\!\left(
       1+\frac{R_\infty^2\sqrt{\nbar}}{\varepsilon}
                \log\frac{h_f}{h_\star}\right).
\end{align*}
\end{proposition}
\begin{proof}
Let $P(\mu,h)$ be the rank-one orthogonal projector onto the ground state of $\Hcal(\mu,h)$, and set $P_j:=P(\mu_j,h_j)$.  By Lemma~\ref{lem:projector-variation}, for $1\le j\le J$ we have
\begin{align*}
    \|P_j-P_{j-1}\| \le\int_{\mu_j}^{\mu_{j-1}}
    \left\|\frac{\d}{\d\mu}P(\mu,\mu)\right\|\,\d\mu \le G(\mu_j)-G(\mu_{j-1})
    =\frac{G(\varepsilon)}J\le\sqrt{\frac23},
\end{align*}
and for $J<j\le L$ we have
\begin{align*}
    \|P_j-P_{j-1}\| \le\int_{h_j}^{h_{j-1}}\|\partial_hP(\varepsilon,h)\|\,\d h \le C_\star\log\frac{h_{j-1}}{h_j}
    =\frac{C_\star R_\star}{J_\star}\le\sqrt{\frac23}.
\end{align*}
The second bound includes the transition from the shared endpoint $j=J$ to the second stage. For rank-one orthogonal projectors,
$1-|\langle\psi_{j-1}|\psi_j\rangle|^2=\|P_j-P_{j-1}\|^2$,
which proves~\eqref{eq:successive_overlap} for every transition. The gap bound follows from Theorem~\ref{thm:spectral-gap-estimate}, because $\tau_u=2R_\infty$. Substituting~\eqref{eq:filter-G} into~\eqref{eq:two-stage-step-counts}, and using $\lceil x\rceil\le x+1$, gives the step counts. Finally, $(\mu_L,h_L)=(\varepsilon,h_\star)$, so Corollary~\ref{cor:final-concentration-scale} applies to $\psi_L$.
\end{proof}

\subsection{Preparation of the initial ground state}
\label{ss:groundstate}
The schedule \eqref{eq:two-stage-schedule} begins at $\Hcal(1,1)$, and the simulation algorithm (discussed subsequently) requires copies of the ground state of $\Hcal(1,1)$. This section discusses how to construct this ground state, by means of a separate Hamiltonian simulation task.
This ground-state preparation can be of independent interest and is cheaper than the other steps of the QCPM algorithm, so we treat it as a standalone subroutine without merging it into the two-stage  schedule~\eqref{eq:two-stage-schedule}.
There may be other, potentially more efficient approaches to achieve this goal; what is presented here is sufficient for our purposes, though it may not be optimal.

We prepare the ground state of $\Hcal(1,1)$ by keeping $h=1$ and interpolating from a quadratic potential $Q$, with known ground state, to the potential $\Phi_1$ used in $\Hcal(1,1)$. Let
$Q(z):=\frac{1}{8R_\infty^2}\|z-e\|^2, W(z):=\Phi_1(z)-Q(z)$.
For $u\in[0,1]$, define
$
    V_u:=(1-u)Q+u\Phi_1=Q+uW,
    \, \Hcal_{\mathrm{in}}(u):=-\tfrac12\Delta+V_u.
$
Thus, $\Hcal_{\mathrm{in}}(0)$ is an isotropic harmonic oscillator and
$\Hcal_{\mathrm{in}}(1)=\Hcal(1,1)$.

\begin{proposition}[Hamiltonians for ground-state construction]
\label{prop:continuous-initialization}
Suppose $\|A\|\le1$, $\|b\|_2\le1$, and $\|c\|_2\le1$. Let $E_{\mathrm{in}}(u)$, $\ket{\chi(u)}$, and
$P(u):=\ket{\chi(u)}\!\bra{\chi(u)}$ denote, respectively, the ground energy, normalized positive ground state, and ground-state
projector of $\Hcal_{\mathrm{in}}(u)$. Then:
\begin{itemize}
    \item[(i)] The ground state of $\Hcal_{\mathrm{in}}(0)$ is 
    \begin{align*}
        \chi(0;z)
        =\left(\frac{1}{2\pi R_\infty}\right)^{\nbar/4}
          \exp\!\left(-\frac{1}{4 R_\infty}\|z-e\|^2\right),
    \end{align*}
    with ground energy $E_{\mathrm{in}}(0)=\nbar/(4R_\infty)$. 
   \item[(ii)] For $u \in [0,1]$, $\Hcal_{\mathrm{in}}(u)$ has a unique ground state and $\mathrm{gap}\bigl(\Hcal_{\mathrm{in}}(u)\bigr) \ge \frac{1}{2R_\infty}$.

    \item[(iii)] There is a computable sequence
    $0=u_0<u_1<\cdots<u_{L_{\mathrm{in}}}=1$ with $L_{\mathrm{in}}
        =\widetilde{\Ocal}\!\left(
            1+\frac{R_\infty^2\sqrt{\nbar}}{\varepsilon}
        \right)$ such that
    \begin{align}\label{eq:init-continuous-overlap}
        |\langle\chi(u_{j-1})|\chi(u_j)\rangle|^2
        \ge\frac{1023}{1024},\qquad j=1,\ldots,L_{\mathrm{in}}.
    \end{align}
\end{itemize}
\end{proposition}

\begin{proof}
By~\eqref{eqn:Phi_hessian_representation},
$\frac{1}{4 R_\infty^2}(I+M^\top M)
\preceq\nabla^2\Phi_1(z)
\preceq\frac{I+M^\top M}{\tau_l^2}$.
Thus
$\nabla^2W=\nabla^2\Phi_1-\frac{1}{4 R_\infty^2} I\succeq\frac{1}{4 R_\infty^2} M^\top M\succeq0$.
Since $W(e)=0$ and $\nabla W(e)=0$, $W$ is nonnegative everywhere and hence
$\nabla^2V_u=\frac{1}{4 R_\infty^2} I+u\nabla^2W\succeq\frac{1}{4 R_\infty^2} I$. Thus, the potential is strongly convex, and the spectral gap formula for Schr\"odinger operators with strongly convex potentials immediately implies (ii). At $u=0$ the Hamiltonian is a standard harmonic oscillator, for which the ground state is known and corresponds to (i).
The upper Hessian bound and
Lemma~\ref{lem:M-Frobenius-bound} also give $
0\preceq\nabla^2W\preceq\frac{4R_\infty^2(1+33\nbar)}{\varepsilon^2}I$ and $ \, \Delta V_u\le \frac{136\nbar R_\infty^2}{\varepsilon^2}.
$

To prove (iii) we characterize the instantaneous velocity of the ground-state projector. Let $\rho_u:=|\chi(u)|^2$ be the distribution induced by the ground state $\ket{\chi(u)}$. Taking derivatives of $\Hcal_{\mathrm{in}}(u) \ket{\chi(u)} = E_{\mathrm{in}}(u) \ket{\chi(u)}$ and using the reduced resolvent gives
$\|P'(u)\|
    \le 2 R_\infty
       \sqrt{\operatorname{Var}_{\rho_u}(W)}.
$
For a real-valued scalar function $f$  with
$\mathbb E_{\rho_u}[f]=0$ such that the pointwise product $f\chi(u)$ lies in the domain of $\Hcal_{\mathrm{in}}(u)$, $f\chi(u)$ is
orthogonal to $\chi(u)$ in $L^2$, because
$\langle\chi(u), f\chi(u)\rangle=\mathbb E_{\rho_u}[f]=0$. The gap bound gives $\mathrm{gap}\bigl(\Hcal_{\mathrm{in}}(u)\bigr)\|f\chi(u)\|^2
\le\langle f\chi(u),
(\Hcal_{\mathrm{in}}(u)-E_{\mathrm{in}}(u))f\chi(u)\rangle$, and the right-hand side is equal to $\frac12\mathbb E_{\rho_u}[\|\nabla f\|^2]$ by expanding $\Hcal_{\mathrm{in}}(u)=-\tfrac12\Delta+V_u$, integrating by parts, and using the ground-state equation $\Hcal_{\mathrm{in}}(u)\ket{\chi(u)}=E_{\mathrm{in}}(u)\ket{\chi(u)}$. Because
$\|f\chi(u)\|^2=\mathbb E_{\rho_u}[f^2]=\operatorname{Var}_{\rho_u}(f)$
and $g_{\mathrm{in}}:=1/(2R_\infty)$ is a lower bound on
$\mathrm{gap}\bigl(\Hcal_{\mathrm{in}}(u)\bigr)$ by item~(ii), we obtain \begin{equation}
\label{eq:poincare}
    \operatorname{Var}_{\rho_u}(f)
    \le R_\infty\mathbb E_{\rho_u}[\|\nabla f\|^2].    
\end{equation}
Applying~\eqref{eq:poincare} to
$W-\mathbb E_{\rho_u}[W]$ and using
$
    \|\nabla W\|^2
    \le
    2\frac{4R_\infty^2(1+33\nbar)}{\varepsilon^2}W,
$
we obtain
$$
    \operatorname{Var}_{\rho_u}(W)
    \le
    2R_\infty
    \frac{4R_\infty^2(1+33\nbar)}{\varepsilon^2}
    \mathbb E_{\rho_u}[W].
$$
By the Hellmann--Feynman identity~\cite{feynman1939forces},
$    E_{\mathrm{in}}'(u)
    =\langle\chi(u),W\chi(u)\rangle
    =\mathbb E_{\rho_u}[W].
$
Moreover, concavity of $E_{\mathrm{in}}$ gives
$E_{\mathrm{in}}'(u)\le E_{\mathrm{in}}'(0)$, and
$
    E_{\mathrm{in}}'(0)
    =\mathbb E_{\rho_0}[W]
    \le
    \frac{R_\infty}{2}
    \frac{136\nbar R_\infty^2}{\varepsilon^2}.
$
Consequently,
$
\operatorname{Var}_{\rho_u}(W)\le 2R_\infty\frac{4R_\infty^2(1+33\nbar)}{\varepsilon^2}E_{\mathrm{in}}'(0)\le R_\infty^2\frac{4R_\infty^2(1+33\nbar)}{\varepsilon^2}\frac{136\nbar R_\infty^2}{\varepsilon^2}.
$
Thus
\begin{align}\label{eq:init-uniform-velocity}
    \|P'(u)\|\le \frac{272R_\infty^4\nbar}{\varepsilon^2}.
\end{align}

We can give a bound on $\|P'(u)\|$ that is much tighter than \eqref{eq:init-uniform-velocity} away from $u=0$; we later combine both bounds. Note that $R_\infty\ge1$ and
$\varepsilon\in(0,1)$ imply $4 R_\infty^2 \cdot \frac{136\nbar R_\infty^2}{\varepsilon^2} =\frac{544\nbar R_\infty^4}{\varepsilon^2} \ge\nbar$, and that $E_{\mathrm{in}}(u)\le E_{\mathrm{in}}(0)+ E_{\mathrm{in}}'(0)$. Using these, we compute the second bound. Lemma~\ref{lem:filter-variance},
inequality~\eqref{eq:poincare} applied to $Q$, and the fact that
$V_u\ge Q$ give
\begin{align*}
    \operatorname{Var}_{\rho_u}(V_u)
    &\le\frac{1}{4}\cdot\frac{136\nbar R_\infty^2}{\varepsilon^2},\notag\\
    \operatorname{Var}_{\rho_u}(Q)
    &\le \frac{1}{2R_\infty} E_{\mathrm{in}}(u)
     \le\frac{\nbar}{8R_\infty^2}+\frac{1}{4}\cdot\frac{136\nbar R_\infty^2}{\varepsilon^2}
     \le\frac{3}{4}\cdot\frac{136\nbar R_\infty^2}{\varepsilon^2}.
\end{align*}
Because $uW=V_u-Q$,
\[
    u^2\operatorname{Var}_{\rho_u}(W)
    \le2\operatorname{Var}_{\rho_u}(V_u)
       +2\operatorname{Var}_{\rho_u}(Q)
    \le2\cdot\frac{136\nbar R_\infty^2}{\varepsilon^2},
\]
and therefore \begin{align} \label{eq:init-logarithmic-velocity}\|P'(u)\|\le\frac{1}{u}\cdot
\frac{8\sqrt{17}\,R_\infty^2\sqrt{\nbar}}{\varepsilon}, \quad u>0.\end{align}

Set
$
    \alpha_{\mathrm{in}}
    :=\min\!\left\{1,\frac{\varepsilon^2}{8704 R_\infty^4\nbar}\right\}.
$
If $\alpha_{\mathrm{in}}=1$, take $u_0=0$, $u_1=1$, and
$L_{\mathrm{in}}=1$. Otherwise, define
$J_{\mathrm{in}}:=\left\lceil
\frac{256\sqrt{17}\,R_\infty^2\sqrt{\nbar}}{\varepsilon}
\log\frac{1}{\alpha_{\mathrm{in}}} \right\rceil$,
$L_{\mathrm{in}}:=J_{\mathrm{in}}+1$,
and set
$u_0:=0$, $u_j:=\alpha_{\mathrm{in}} \exp\!\left(\frac{(j-1)\log\frac{1}{\alpha_{\mathrm{in}}}}{J_{\mathrm{in}}}\right)$ for $j=1,\ldots,L_{\mathrm{in}}$. For $u \in [0,\alpha_{\mathrm{in}}]$ we use \eqref{eq:init-uniform-velocity} and obtain $\int_0^{\alpha_{\mathrm{in}}}\|P'(u)\|\,du
\le\alpha_{\mathrm{in}}\cdot
\frac{272R_\infty^4\nbar}{\varepsilon^2}
\le1/32$ by the choice of $\alpha_{\mathrm{in}}$. For $u > \alpha_{\mathrm{in}}$ we use \eqref{eq:init-logarithmic-velocity}, and obtain, for
$j=2,\ldots,L_{\mathrm{in}}$,
\begin{equation*}
    \|P(u_j)-P(u_{j-1})\|
    \le\int_{u_{j-1}}^{u_j}
      \frac{8\sqrt{17}\,R_\infty^2\sqrt{\nbar}}{\varepsilon\,u}\,du
    =
    \frac{8\sqrt{17}\,R_\infty^2\sqrt{\nbar}\log(1/\alpha_{\mathrm{in}})}{\varepsilon\,J_{\mathrm{in}}}
    \le\frac{1}{32},
\end{equation*}
by the choice of $J_{\mathrm{in}}$. In summary,
$\|P(u_1)-P(u_0)\|\le\frac{1}{32}$, $\|P(u_j)-P(u_{j-1})\|\le\frac{1}{32}$ for $j\ge2$.
Because
$1-|\langle\chi(u)|\chi(v)\rangle|^2=\|P(u)-P(v)\|^2$,
this proves~\eqref{eq:init-continuous-overlap}. Finally,
\begin{align*}
    \frac{8\sqrt{17}\,R_\infty^2\sqrt{\nbar}}{\varepsilon}
    <33\,\frac{R_\infty^2\sqrt{\nbar}}{\varepsilon},
    \quad 
    L_{\mathrm{in}}
    \le 2+\frac{256\sqrt{17}\,R_\infty^2\sqrt{\nbar}}{\varepsilon}
      \log\!\Bigl(\max\Bigl\{1,\,
      \frac{8704R_\infty^4\nbar}{\varepsilon^2}
      \Bigr\}\Bigr).
\end{align*}
which concludes the proof of (iii).
\end{proof}

\subsection{Spatial discretization and Hamiltonian simulation}\label{ss:discretization-block-encoding}

Consider the parameters $(\mu_j,h_j)$ along the schedule described in Section~\ref{ss:Qsim}. Thus, $\mu_j\in[\varepsilon,1]$ and $0<h_j\le1$. In this section we first discuss simulation of the $j$-th Hamiltonian $\Hcal_j = \Hcal(\mu_j,h_j)$, and analyze the corresponding cost. Then we apply a similar analysis to the preparation of the ground state of $\Hcal(1,1)$, which, as mentioned in Section~\ref{ss:groundstate}, is an easier time evolution task. We discretize the (continuous) Hamiltonian on the periodic box $\Omega=[0,1]^{\nbar}$ using $N$ grid points per dimension.
Our discretization scheme follows a similar treatment in~\cite{li2023efficient,leng2026operator} for block-encoding pseudo-differential operators, and this is usually sufficient to choose a discretization number $N$ that is polylogarithmic in $\nbar$ and $1/\varepsilon$. We formally summarize the cost of spatial discretization as an assumption, emphasizing that similar assumptions are common when translating continuous-space operators to their discrete counterpart for digital simulation, and they are validated by numerical practice. 

\begin{assumption}[High-precision spatial discretization]\label{a:grid}
    We assume that choosing some 
    $$N=\polylog(\nbar,1/\varepsilon)$$ 
    guarantees the following properties after spatial discretization $z_k=k/N$ for $k\in\{0,\ldots,N-1\}^{\nbar}$: (i) each discrete Hamiltonian $\Hcal_j$ for the schedule in Section~\ref{ss:Qsim} has a unique ground state $\ket{\psi_j}$, (ii) its spectral gap is at least $h_j/(2R_\infty)$, and $|\langle\psi_{j-1}|\psi_j\rangle|^2\ge 1-\frac1{256}$ for all $j=1,\dots, L$;
   (iii) measuring $\ket{\psi_L}$ in the position basis and decoding the grid point yields $z>0$, $s(z)>0$, and $d_2(z,s(z);\varepsilon)\le\gamma\varepsilon$ with probability at least $2/3$. Properties (i) and (ii) also hold for the ground-state preparation Hamiltonians of Proposition~\ref{prop:continuous-initialization} with the same $N$.
\end{assumption}

Let $z_k=k/N$ for $k\in\{0,\ldots,N-1\}^{\nbar}$, and let $F=\mathrm{QFT}_N^{\otimes \nbar}$ be the quantum Fourier transform on the $\nbar\log_2N$ system qubits. We write the Fourier discretization of the $j$-th Hamiltonian as $\Hcal_j=A_j+B_j$, where $A_j$ can be diagonalized by the quantum Fourier transform $A_j := F^\dagger D_j F$, and the diagonal terms in the kinetic and potential parts are
\begin{align*}
    D_j\ket\xi =2\pi^2h_j^2\|\xi\|^2\ket\xi,\qquad B_j\ket k=\Phi_{\mu_j}(z_k)\ket k,
\end{align*}
where $\xi\in\{-N/2,\ldots,N/2-1\}^{\nbar}$. In the above equation, $A_j$ denotes the kinetic operator, and $B_j$ denotes the potential; we use this notation only in this section for consistency with \cite{low2018interaction}, which we use later for Hamiltonian evolution. (Note that unsubscripted $A$ continues to denote the constraint matrix of the LO problem.) Evolution of the potential function is diagonal in the computational basis
\begin{align}\label{eq:diagonal-potential-evolution}
    e^{-itB_j}\ket k=e^{-it\Phi_{\mu_j}(z_k)}\ket k.
\end{align}
This operator can be implemented in an inexpensive manner by reversibly computing the potential, applying the phase, and uncomputing; see, e.g., \cite[Ch.~6]{nannicini2025quantum}. The same construction accepts a binary time register in superposition, by computing the product of the time and potential values before applying the phase. On the other hand, by definition of the matrix exponential, $e^{-itA_j}=F^\dagger e^{-itD_j}F$, so the evolution of the kinetic operator is diagonal in the Fourier basis, and it can be implemented in a similar way as~\eqref{eq:diagonal-potential-evolution} by conjugating with the QFT. Direct calculations show that
$
    \|A_j\|\le a_j:=\frac{\pi^2h_j^2N^2\nbar}{2}
$ and this quantity appears in the cost of the simulation algorithm.

For the normalized LO data used in Section~\ref{ss:Qsim}, the Hessian bound~\eqref{eqn:Phi_hessian_representation} and Lemma~\ref{lem:M-Frobenius-bound} give
\[
    0\le\Phi_\mu(z)\le
    \frac{(1+33\nbar)\nbar R_\infty^2}{2\tau_l^2}
    =\frac{2(1+33\nbar)\nbar R_\infty^4}{\varepsilon^2},
    \qquad z\in\Omega,\quad \mu\in[\varepsilon,1].
\]
Indeed, $\Phi_\mu$ and its gradient vanish at $z(\mu)$, and $\|z-z(\mu)\|^2\le \nbar R_\infty^2$. Combining the normalizations of the two Hamiltonians, and following Eq.~\eqref{eq:bks-evolution-oracles} as well as \cite{low2018interaction}, we set the evolution time $\tau$ and Hamiltonian bound $\Gamma$ as
$\tau:=\frac{\pi}{2\Gamma}, \quad \Gamma:=\frac{\pi^2N^2\nbar}{2}
       +\frac{2(1+33\nbar)\nbar R_\infty^4}{\varepsilon^2}.
$

Let $U_{A,j}$ be a block-encoding of $A_j/a_j$, and let $O_{\Phi,j}$ be a reversible circuit evaluating $\Phi_{\mu_j}(z_k)$ to the working precision. $U_{A,j}$ can be constructed as discussed above with  $\widetilde{\Ocal}(\nbar\log^2N)$ gates. $O_{\Phi,j}$ can be constructed by evaluating $s(z)=Mz+q$ and the clipped logarithms gives with cost $\widetilde{\Ocal}(\nbar+\mathrm{nnz}(A))$ gates. (We absorb the required arithmetic and rotation precision in the polylogarithmic factors.) As is standard in the literature, by ``implementing a unitary $V$ up to error $\eta$'' we mean providing a circuit for a unitary that has distance at most $\eta$ from $V$ in the operator norm, see, e.g., \cite{nannicini2025quantum}.

\begin{lemma}[Hamiltonian simulation]\label{lem:ip-hamiltonian-simulation}
Assume $\nnz(A)\ge m+n$. For every $j$, $t\in\mathbb R$, and $\eta\in(0,1/2)$, the evolution $e^{-it\Hcal_j}$ can be implemented with error $\eta$ using $\widetilde{\Ocal}_{1/\varepsilon,1/\eta}(1+\nbar h_j^2|t|)$ queries to the potential oracle $O_{\Phi,j}$, including inverse and controlled calls, and $\widetilde{\Ocal}_{1/\varepsilon,1/\eta}\!\left[
          \nnz(A)
          \left(1+\nbar h_j^2|t|\right)\right]$
elementary gates. The same bounds hold for controlled evolutions and circuit inverses.
\end{lemma}
\begin{proof}
We use the interaction picture with respect to the potential $B_j$, see, e.g., \cite{low2018interaction}. The corresponding Hamiltonian evolution can be written as
\begin{align*}
    e^{-it\Hcal_j}=e^{-itB_j}\,
    \mathcal T\exp\!\left(-i\int_0^t A_{j,I}(s)\,\mathrm ds\right),  \text{ where } A_{j,I}(s)=e^{isB_j}A_je^{-isB_j} \text{ and }
    \|A_{j,I}(s)\|\le a_j.
\end{align*}
We can construct a block-encoding of $A_{j,I}(1)$ by applying $U_{A,j}$ conjugated with the time-controlled potential evolution; the latter uses a constant number of calls to $O_{\Phi,j}$. To obtain a block-encoding of $A_{j,I}(s)$, we multiply $O_{\Phi,j}$ with a binary time register using standard arithmetic, and the gate cost is polynomial in the register length. 
Because $\left\|\frac{\mathrm d}{\mathrm ds}A_{j,I}(s)\right\|\le2\|B_j\|a_j\le2\Gamma a_j$, 
using the truncated Dyson-series algorithm of \cite{low2018interaction} we can assume that the length of the time register depends logarithmically on $\Gamma$, the evolution time, and the inverse accuracy.

Apply \cite[Lemma~6]{low2018interaction}, choosing $B_j$ (i.e., the potential term) to be the directly-simulated term, and $A_j$ (i.e., the kinetic term) to be the block-encoded term. For $t>0$, we split the evolution into $q=\max\{1,\lceil2a_jt\rceil\}$ segments. Each segment uses a polylogarithmic number of calls to the block-encoding of $A_{j,I}$ and the evolution of $B_j$, each of which requires a constant number of evaluations and uncomputations of the potential. Thus the total number of queries to $O_{\Phi,j}$, counting inverses and controlled versions, is $\widetilde{\Ocal}_{1/\varepsilon,1/\eta}(1+a_jt)$. To approximate the final evolution up to error $\eta$, it suffices to construct each segment with error $\Ocal(\eta/q)$.
The potential evaluations, kinetic block-encodings, and auxiliary arithmetic within each segment cost $\widetilde{\Ocal}_{|t|,1/\varepsilon,1/\eta}(\nbar\log^2N+\nnz(A))=\widetilde{\Ocal}_{|t|,1/\varepsilon,1/\eta}(\nnz(A))$ gates, because $N=\polylog(\nbar,1/\varepsilon)$ and $\nbar=m+n+2=\Ocal(\nnz(A))$ under the stated assumption. Substituting $a_j=\pi^2h_j^2N^2\nbar/2$ gives the claimed query and gate bounds.
\end{proof}

\begin{lemma}[Phase estimation and reflection oracles]\label{lem:bks-oracle-cost}
Assume $\nnz(A)\ge m+n$. Suppose $\Hcal_j$ has a unique ground state $\ket{\psi_j}$ and its spectral gap is lower bounded by a known value $\Delta_j>0$. Let $U_j$, $\varphi_j$, and $\vartheta_j=\tau\Delta_j$ be as in~\eqref{eq:bks-evolution-oracles}. The oracle $\mathrm{PE}(U_j,\vartheta_j/5)$ (recall Section~\ref{ss:traversalalg}) can be implemented with error $\eta\in(0,1/2)$ using $\widetilde{\Ocal}_{1/\varepsilon,1/\eta}(1+\nbar h_j^2/\Delta_j)$ queries to the potential oracle $O_{\Phi,j}$, including inverse and controlled calls, and
\begin{align}\label{eq:bks-oracle-implementation-cost}
    \widetilde{\Ocal}_{1/\varepsilon,1/\eta}\!\left[
          \nnz(A)
          \left(1+\frac{\nbar h_j^2}{\Delta_j}\right)\right]
\end{align}
elementary gates. Given $\tilde\varphi_j$ with $|\tilde\varphi_j-\varphi_j|\le\vartheta_j/5$, the reflection $R_j=I-2\ket{\psi_j}\!\bra{\psi_j}$ can be implemented with the same query and gate bounds and error tolerance.
\end{lemma}
\begin{proof}
For phase estimation with precision
$\vartheta_j/5$ we directly follow the implementation discussed in \cite[Definition~III.1]{boixo2010fast}. For the reflection, we follow \cite[Definitions~III.2--III.3]{boixo2010fast}. 

Choose $p$ so that $2^p=\Theta(1/\vartheta_j)$, with a sufficiently large constant. By \eqref{eq:bks-evolution-oracles}, each repetition of phase estimation then uses $p=\Ocal(\log(2+\Gamma/\Delta_j))$ controlled powers
\[
    U_j^{2^k}=e^{-i(2^k\tau)\Hcal_j},\qquad k=0,\ldots,p-1,
\]
whose total evolution time is
\[
    \sum_{k=0}^{p-1}2^k\tau
       =\Ocal\!\left(\frac{\tau}{\vartheta_j}\right)
       =\Ocal(1/\Delta_j).
\]
We can implement each power directly using Lemma~\ref{lem:ip-hamiltonian-simulation}. Because the number of repetitions is logarithmic, the number of queries to $O_{\Phi,j}$, including inverses and controlled versions, is
\[
    \widetilde{\Ocal}_{1/\varepsilon,1/\eta}\!\left(p+\frac{a_j}{\Delta_j}\right)
       =\widetilde{\Ocal}_{1/\varepsilon,1/\eta}\!\left(1+\frac{\nbar h_j^2}{\Delta_j}\right),
\]
and implementing these potential queries and the accompanying kinetic block-encodings and arithmetic costs $\widetilde{\Ocal}_{
1/\varepsilon,1/\eta,1/\Delta_j}(\nnz(A))$ gates per query under the assumption $\nnz(A)\ge m+n$, giving~\eqref{eq:bks-oracle-implementation-cost}.
\end{proof}

By Theorem~\ref{thm:spectral-gap-estimate}, we can lower bound the gaps as $\Delta_j=h_j/(2R_\infty)$. Under the assumption $\nnz(A)\ge m+n$, Lemma~\ref{lem:bks-oracle-cost} therefore gives $\widetilde{\Ocal}_{1/\varepsilon,1/\eta}(1+R_\infty\nbar h_j)$ potential-oracle queries and
$\widetilde{\Ocal}_{1/\varepsilon,1/\eta}\!\left[
          \nnz(A)(1+R_\infty \nbar h_j)\right]$
elementary gates for each phase-estimation or reflection call in Theorem~\ref{thm:bks-eigenpath-traversal}.

\begin{corollary}[Cost of preparing the initial state]
\label{cor:constructive-initial-state-preparation}
Suppose Assumption~\ref{a:grid} holds, and
$\nnz(A)\ge m+n$. For every
$\eta\in(0,1/2)$, the ground state $\ket{\psi_0}$ of the discretized Hamiltonian
$\Hcal(1,1)$ can be prepared to trace-distance error at most $\eta$
using
\begin{align}\label{eq:initialization-query-cost}
    \widetilde{\Ocal}_{\frac{1}{\eta}}\!\left[
        \left(1+\frac{R_\infty^2\sqrt{\nbar}}{\varepsilon}\right)
        (1+R_\infty\nbar)
    \right]
\end{align}
queries to the potential function, and the same quantity multiplied by
$\widetilde{\Ocal}(\nnz(A))$ elementary gates. 
\end{corollary}

\begin{proof}
An oracle for $V_u$ uses one call to an oracle for the potential $\Phi_1$ plus
$\widetilde{\Ocal}_{{\frac{1}{\eta}},\,\frac{1}{\varepsilon}}(\nbar)$ arithmetic gates. Because
$0\le V_u\le\Phi_1$, we can use the simulation result in Lemma~\ref{lem:ip-hamiltonian-simulation}, with $h=1$. Lemma~\ref{lem:bks-oracle-cost} then shows that we can implement phase estimation and reflections with
$\widetilde{\Ocal}_{\frac{1}{\eta},\,\frac{1}{\varepsilon}}(1+R_\infty\nbar)$ potential queries and
$\widetilde{\Ocal}_{\frac{1}{\eta},\,\frac{1}{\varepsilon}}(\nnz(A)(1+R_\infty\nbar))$ gates. We then apply Theorem~\ref{thm:bks-eigenpath-traversal} directly. As discussed in Section~\ref{ss:groundstate}, the initial state is an isotropic Gaussian state, which can be prepared efficiently with the Grover-Rudolph construction \cite{grover2002creating}. Thus, the total query complexity is
$\widetilde{\Ocal}_{\frac{1}{\eta},\,\frac{1}{\varepsilon}}(L_{\mathrm{in}}(1+R_\infty\nbar))$;
\eqref{eq:initialization-query-cost} follows from
Proposition~\ref{prop:continuous-initialization} (iii), and the gate bound follows similarly.
\end{proof}

\subsection{Quantum central path algorithm}\label{ss:QCPM}
We combine the traversal theorem in Section~\ref{ss:traversalalg} with the schedules of Section~\ref{ss:Qsim} and the circuit constructions of Section~\ref{ss:discretization-block-encoding}. As in Section~\ref{ss:Qsim}, the LO data are normalized, $\nbar=n+m+2$, and $R_\infty$ and $\gamma$ are fixed. 
The required initial discretized ground
state can be prepared constructively using the initialization path of
Proposition~\ref{prop:continuous-initialization} with the complexity shown in
Corollary~\ref{cor:constructive-initial-state-preparation}.

\begin{theorem}\label{t:QIPMComplexity}
    Let $\varepsilon \in (0,1)$. Choose $\gamma \in (0,1)$ and $\delta \in (0,1)$. Suppose Assumption~\ref{a:grid} holds, and $\nnz(A) \geq m+n$. Then, with probability at least $1-\delta$, a quantum algorithm returns a classical vector $z \in \R{\nbar}_{++}$ such that $(z, s(z)) \in \Ncal_2 (\gamma)$ with $d_2(z,s(z);\varepsilon) \le \gamma \varepsilon$.
    The algorithm uses $\widetilde{\Ocal}_{1/\delta}\!\left(\nbar^{3/2}/\varepsilon+\sqrt{\nbar}/\varepsilon^{3/2}\right)$ queries in total to the reversible evaluation oracles $O_{\Phi,j}$ for the potentials $\Phi_{\mu_j}$ along the schedule, including inverse and controlled calls.
    Moreover, if the LO problem data $(A, b, c)$ is specified in a classical data structure (e.g., a sparse matrix list, etc.) such that $\|A\|, \|b\|,\|c\|\le 1$, this quantum algorithm can be implemented with
    \begin{align}\label{eq:qcpm-gate-complexity}
        \widetilde{\mathcal{O}}_{\nbar, \frac{1}{\delta}, \frac{1}{\varepsilon}}\Big(\left(\nbar^{3/2}\,\varepsilon^{-1} + \nbar^{1/2}\,\varepsilon^{-3/2}\right)\cdot \nnz(A)\Big)
        \end{align}
    elementary gates, where $\nnz(A)$ is the number of non-zeros in the matrix $A$.
\end{theorem}

For $\mathrm{nnz}(A)=\Ocal(\nbar)$ and fixed $\varepsilon$, this amounts $\widetilde{\Ocal}_{1/\delta}(\nbar^{5/2})$ gates.

\begin{proof}[Proof of Theorem~\ref{t:QIPMComplexity}]
Use the two-stage schedule of Proposition~\ref{prop:two-stage-schedule}, with the refined step counts
\begin{align*}
    J=\lceil16G(\varepsilon)\rceil,\qquad
    J_\star=\lceil16C_\star R_\star\rceil.
\end{align*}
Index the path by $\ell=0,\ldots,L$, where $L=J+J_\star$. The projector estimates in the proof of Proposition~\ref{prop:two-stage-schedule} give successive projector differences at most $1/16$. Thus the corresponding ground states satisfy
\[
    |\langle\psi_{\ell-1}|\psi_\ell\rangle|^2 \ge1-\frac1{256}>\frac56.
\]
Assumption~\ref{a:grid} supplies these estimates for the discretized states, so Theorem~\ref{thm:bks-eigenpath-traversal} applies with $\beta=1/8$. 
The same assumptions give gap lower bounds $\Delta_\ell=h_\ell/(2R_\infty)$. The step counts satisfy $L=\widetilde{\Ocal}\!\left(\frac{\sqrt{\nbar}}{\varepsilon^{3/2}}\right)$, because $G(\varepsilon)=\widetilde{\Ocal}(\sqrt{\nbar}/\varepsilon^{3/2})$, $C_\star=\Ocal(\sqrt{\nbar}/\varepsilon)$, and $R_\star=\Ocal(\log \nbar)$.

Prepare $r:=\left\lceil64\log\frac{64(L+1)}{\delta}\right\rceil$ copies of $\ket{\psi_0}$ using Corollary~\ref{cor:constructive-initial-state-preparation}, with the per-copy error specified below, and obtain the initial phase estimate by one call to $\mathrm{PE}(U_0,\vartheta_0/5)$. Apply the BKS traversal to these registers along the prescribed path. Its intrinsic joint error amplitude is at most $Le^{-r/64}\le\delta/64$.

To bound the running time, cap the number of reflections at each transition at
\[
    K:=\left\lceil3r+ \frac{2}{\log(7/4)}\log\frac{64(L+1)}{\delta}\right\rceil,
\]
and return failure if the cap is exceeded. If $M_\ell$ denotes the reflection count at one transition, the one-step bound of \cite[Lemma~IV.8]{boixo2010fast} gives
\[
    \mathbb P(M_\ell>K) \le(7/4)^{3r-K} \le\left(\frac{\delta}{64(L+1)}\right)^2.
\]
Truncating a branch of probability $p$ changes its coherent implementation by amplitude at most $2\sqrt p$, so these cutoffs contribute total error amplitude at most $\delta/32$. There are at most $M:=1+L(2r+K)$ phase-estimation and reflection calls, including the initial phase estimation. Implement each to error $\delta/(64M)$ using Lemma~\ref{lem:bks-oracle-cost}, and prepare each initial copy to trace-distance error $\delta/(64r)$ using Corollary~\ref{cor:constructive-initial-state-preparation}. Summing together all the errors, this gives trace distance at most $\delta/2$ from the ideal joint output of $r$ final ground states.

We first count the queries to the potential function using the individual $h_\ell$. By Lemma~\ref{lem:bks-oracle-cost}, each phase-estimation or reflection call at index $\ell$ uses $\widetilde{\Ocal}_{1/\varepsilon,1/\delta}\!\left(1+\frac{\nbar h_\ell^2}{\Delta_\ell}\right)=\widetilde{\Ocal}_{1/\varepsilon,1/\delta}(1+R_\infty\nbar h_\ell)$ queries to $O_{\Phi,\ell}$, including inverse and controlled calls. Each transition uses only the source and target oracles, with at most $2r+K$ phase-estimation and reflection calls in total. Because $2r+K$ is polylogarithmic, summing over all transitions and including the initial phase estimation gives a total potential-query count of
\begin{align}\label{eq:qcpm-weighted-queries}
    \widetilde{\Ocal}_{1/\varepsilon,1/\delta}\!\left(
        L+1+R_\infty\nbar\sum_{\ell=0}^Lh_\ell
    \right).
\end{align}
For the gate cost, the circuit constructions in Section~\ref{ss:discretization-block-encoding} implement each potential query together with the accompanying kinetic operations and arithmetic using $\widetilde{\Ocal}_{1/\varepsilon,1/\delta}(\nbar+\mathrm{nnz}(A))$ gates. Thus each phase-estimation or reflection call at index $\ell$ uses $\widetilde{\Ocal}_{1/\varepsilon,1/\delta}\!\left[\bigl(\nbar+\mathrm{nnz}(A)\bigr)(1+R_\infty \nbar h_\ell)\right]$
gates, and the same summation gives
\begin{align}\label{eq:qcpm-weighted-gates}
    \widetilde{\Ocal}_{1/\varepsilon,1/\delta}\!\left[
       \bigl(\nbar+\mathrm{nnz}(A)\bigr)
       \left(L+1+R_\infty \nbar\sum_{\ell=0}^Lh_\ell\right)
    \right].
\end{align}
The reversible phase processing and register permutations in BKS cost $\widetilde{\Ocal}_{1/\varepsilon,1/\delta}(\nbar)$ gates per transition and are absorbed in this bound. The gate costs of Lemma~\ref{lem:bks-oracle-cost} already include the controlled evolutions supplied by Lemma~\ref{lem:ip-hamiltonian-simulation}.

To upper bound \eqref{eq:qcpm-weighted-gates} we estimate the cost of the summation over $h_\ell$; we split this estimate over the two stages of the schedule. For the first stage of the schedule, define
\begin{align*}
    I :=\int_\varepsilon^1 \mu\bigl(-G'(\mu)\bigr)\,\mathrm d\mu =4R_\infty^{3/2}\sqrt{\nbar}
           (\varepsilon^{-1/2}-1)
       +\frac{24R_\infty^2\sqrt{\nbar}}{\varepsilon}(1-\varepsilon).
\end{align*}
By definition of the first-stage schedule~\eqref{eq:two-stage-schedule}, we can write $h_j=G^{-1}(u_j)$, where the intermediate mesh points are $u_j:=jG(\varepsilon)/J$ for $j=0,\ldots,J$, with spacing $\Delta u:=G(\varepsilon)/J$. 
Using the substitution $u = G(\mu)$ in the change-of-variables formula, we have (note that we are switching the endpoints of the integral, which introduces a minus sign)
\[
    \int_0^{G(\varepsilon)}G^{-1}(u)\,\mathrm du =\int_\varepsilon^1\mu\bigl(-G'(\mu)\bigr)\,\mathrm d\mu=I.
\]
Since $G^{-1}$ is decreasing, $G^{-1}(u)\ge h_j$ for $u\in[u_{j-1},u_j]$. Integrating this inequality over each interval and summing over $j=1,\ldots,J$ yields
\[
    \Delta u\sum_{j=1}^J h_j
    \le\sum_{j=1}^J\int_{u_{j-1}}^{u_j}G^{-1}(u)\,\mathrm du
    =I.
\]
Dividing by $\Delta u$ and adding $h_0=1$ therefore gives
\[
    \sum_{j=0}^J h_j
    \le1+\frac{J}{G(\varepsilon)} I
    \le2+16I,
\]
where we used $J\le16 G(\varepsilon)+1$ and $I\le  G(\varepsilon)$. For the second stage of the schedule, excluding the initial point that we already accounted for, the same estimate applied to $h_fe^{-u}$ gives
\[
    \sum_{k=1}^{J_\star}h_{J+k}
       \le\frac{J_\star}{R_\star}h_f(1-e^{-R_\star})
       \le h_f(16C_\star+1).
\]
Because $h_f=\varepsilon$ and $C_\star=24R_\infty^2\sqrt{\nbar}/\varepsilon$, we conclude that $\sum_{\ell=0}^Lh_\ell=\Ocal(\sqrt{\nbar}/\varepsilon)$. Substituting this upper bound and $L=\widetilde{\Ocal}\!\left(\frac{\sqrt{\nbar}}{\varepsilon^{3/2}}\right)$ into~\eqref{eq:qcpm-weighted-queries}, with $R_\infty$ fixed, proves the stated bound on the number of queries to the potential. Substituting the same estimates into~\eqref{eq:qcpm-weighted-gates} gives~\eqref{eq:qcpm-gate-complexity}, because $\nbar+\mathrm{nnz}(A)=\Ocal(\mathrm{nnz}(A))$ under the assumption $\mathrm{nnz}(A)\ge m+n$. The schedule parameters are computed by scalar bisection and elementary arithmetic to the working precision; their approximation errors are included in the simulation error budget. This takes $\widetilde{\Ocal}_{1/\varepsilon,1/\delta}(L+1)$ operations and does not change either bound.

Finally, $(\mu_L,h_L)=(\varepsilon,h_\star)$. The concentration estimate underlying Corollary~\ref{cor:final-concentration-scale}, transferred to the spatially-discretized state, by assumption gives probability at least $2/3$ that an ideal position sample satisfies $z>0$, $s(z)>0$, and $d_2(z;\varepsilon)\le\gamma\varepsilon$. Measure the $r$ output registers, check these conditions, and return the first verified point. For the ideal product state, all checks fail with probability at most $3^{-r}$. The joint trace-distance bound, including the failure flag, therefore gives total failure probability at most
$3^{-r}+\delta/2\le\delta$. By Corollary~\ref{cor:constructive-initial-state-preparation},
preparing the $r$ initial copies with trace-distance error
$\delta/(64r)$ has elementrary gates complexity
\[
    r\,\widetilde{\Ocal}_{\frac{r}{\delta}}\!\left[
        \nnz(A)
        \left(1+\frac{R_\infty^2\sqrt{\nbar}}{\varepsilon}\right)
        (1+R_\infty\nbar)
    \right]
    =
    \widetilde{\Ocal}_{\frac{1}{\delta}}\!\left[
        \nnz(A)
        \left(1+\frac{R_\infty^2\sqrt{\nbar}}{\varepsilon}\right)
        (1+R_\infty\nbar)
    \right],
\] where the equality uses
$r=\Ocal(\log((L+1)/\delta))$. Checking the final samples costs $\widetilde{\Ocal}_{1/\varepsilon}(r(\nbar+\mathrm{nnz}(A)))$ operations. Both are absorbed in~\eqref{eq:qcpm-gate-complexity}.
\end{proof}

\section{Numerical demonstration}\label{s:numerical}
In this section, we present a numerical example to illustrate and visualize the concept of the quantum central path, as discussed in Section~\ref{s:QIPM}.
We consider a primal optimization problem~\eqref{e:LP} determined by the following constraint matrix $A$, and vector $b$, and $c$:
\begin{align*}
    A = \begin{bmatrix}
        1 \\ 0
    \end{bmatrix},\quad 
    b = \begin{bmatrix}
        1\\-1
    \end{bmatrix},\quad 
    c = \begin{bmatrix}
        2
    \end{bmatrix}.
\end{align*}
According to \eqref{eqn:M_and_r}, the self-dual embedding model~\eqref{e:self-dual} is determined by the matrix $M \in \R{5\times 5}$ and vector $q \in \R{5}$:
\begin{align*}
    M = \begin{bmatrix}
       0 & 0 & 1 & -1 & 1\\
       0 & 0 & 0 & 1 & 0\\
       -1 & 0 & 0 & 2 & 0\\
       1 & -1 & -2 & 0 & 3\\
       -1 & 0 & 0 & -3 & 0
    \end{bmatrix},\quad 
    q = \begin{bmatrix}
        0\\0\\0\\0\\5
    \end{bmatrix}.
\end{align*}
It can be readily verified that the limiting point of the central path is~\cite[Example I.31]{roos2005interior}: $z = [1.6, 0, 0.8, 0.8, 0]^\top$, $s(z) = [0, 0.8, 0, 0, 1]^\top$.
Due to our construction, the quantum central path is encoded in the ground state of the central-path Hamiltonian, which is a Schr\"odinger operator of the following form:

\begin{align*}
    \Hcal(\mu, h) = - \frac{h^2}{2}\nabla^2 + \Phi_\mu(z),\qquad
    \Phi_\mu(z)=\frac{q^\top z}{\mu}
      +\sum_{i=1}^{5}\phi(z_i)+\sum_{i=1}^{5}\phi(s_i(z))
      -5(1-\log\mu).
\end{align*}
Here we set $h=\mu$, $s(z)=Mz+q$, and $\phi$ is the quadratically continued logarithm in~\eqref{eq:clipped-log-definition}. We set $\varepsilon=0.1$ and $R_\infty=2$, which bounds all central-path coordinates and slacks in this example for $\mu\in[0.1,1]$. Thus, the clipping thresholds $\tau_l=0.025$ and $\tau_u=4$ are fixed.

Note that this operator is an unbounded operator acting on functions in $\R{5}$.
A tensor-product discretization with $N$ points per dimension produces a matrix of size $N^5\times N^5$. To reduce the computational overhead and make a more intuitive graphical demonstration, we restrict this problem to the first two dimensions by substituting the last three components of the $\mu$-center $z(\mu)$ into the potential field $\Phi_\mu(z)$ for $\mu\in[0.1,1]$. This substitution effectively reduces the central-path Hamiltonian to a 2-dimensional Schr\"odinger operator, whose ground state can be numerically computed efficiently using a personal laptop.
We discretize the reduced operator on $[-0.5,2.5]\times[-0.5,3]$ with homogeneous Dirichlet boundary conditions and a five-point finite-difference Laplacian of mesh spacing $0.01$. We compute its lowest eigenpair and normalize the squared modulus of the eigenvector so that its grid integral equals one.

In Figure~\ref{fig:numerical}, we illustrate the normalized ground-state probability density $|\psi_\mu|^2$ of the reduced (i.e., 2-dimensional) Schr\"odinger operator for $\mu \in \{1, 0.8, 0.5, 0.2, 0.1\}$. We choose $h=\mu$ to illustrate concentration during the first (i.e., transport) stage. As $\mu$ decreases, the density narrows near the central path, depicted as a white dashed curve. At finite $h$, note that the density's mean and mode need not coincide exactly with the minimizer of the potential $\Phi_\mu$.

\begin{figure}[H]
        \centering
        \includegraphics[width=\linewidth]{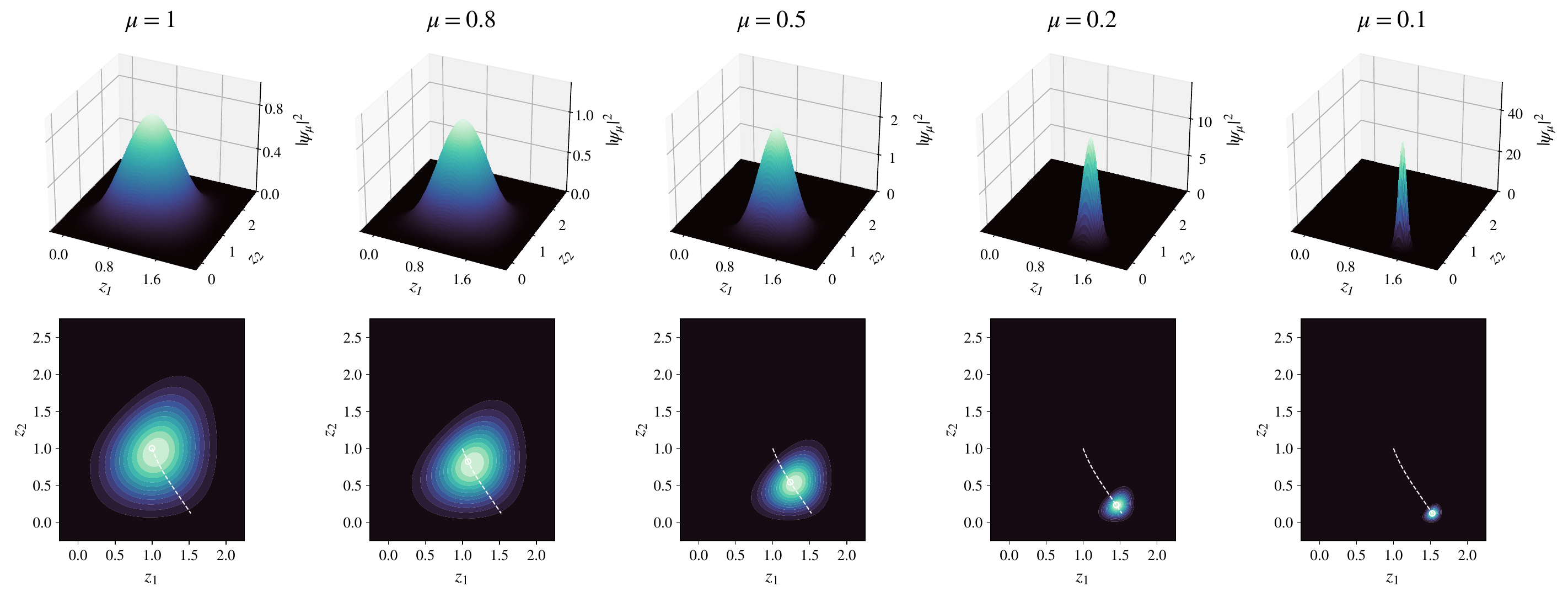}
        \caption{Ground-state probability densities $|\psi_\mu|^2$ of the two-dimensional restriction of the log-barrier Hamiltonian~\eqref{eqn:Hcal_mu}, with $h=\mu$. The top row shows three-dimensional surfaces and the bottom row shows heat maps of the same densities. White dashed curves show the projected central path, and circles mark its point at the displayed $\mu$.}
        \label{fig:numerical}
\end{figure}

\section{Conclusion and outlook}\label{s:con}
In this work, we proposed a new quantum algorithm that solves linear optimization problems by quantum evolution of the central path. With a two-stage quantum eigenpath traversal schedule, one can obtain an $\varepsilon$-optimal solution to a linear optimization problem involving $m$ constraints and $n$ variables using at most
\[\widetilde{\mathcal{O}}\Big(\left((m+n)^{3/2}\,\varepsilon^{-1} + (m+n)^{1/2}\,\varepsilon^{-3/2}\right)\cdot \nnz(A)\Big)\]
elementary gates, where $\nnz(A)$ is the total number of nonzero entries found in $A$. This bound avoids QRAM and condition-number dependence.

Though we leave this for future work, our framework should readily generalize to more complex classes of convex optimization problems, such as semidefinite optimization and second-order conic optimization. Our work highlights a previously unexplored connection between interior point methods and quantum Hamiltonian simulation, and we believe further exploration into this relationship to be a worthy endeavor.

\subsection*{Acknowledgments}
The authors would like to thank the anonymous referees for useful comments. 
JL is partially supported by the Simons Quantum Postdoctoral Fellowship, DOE QSA grant \#FP00010905, a Simons Investigator award through Grant No. 825053, and the Commonwealth Cyber Initiative (CCI) Faculty Fellow program. 
JB is supported by ONR award \#N000142312585.
GN is partially supported by ONR award \#N000142312585. 
TT is partially supported by the DARPA grant 
W911NF2010022: The Quantum Computing Revolution and Optimization: Challenges and Opportunities; 
and by the NSF Grant No. 2128527. 
XW is partially supported by NSF CAREER Award CCF-1942837, a Sloan research fellowship, and the U.S. Department of Energy, Office of Science, Accelerated Research in Quantum Computing, Fundamental Algorithmic Research toward Quantum Utility (FAR-Qu).

\bibliographystyle{alpha} 
\bibliography{sdpbib}

\newpage
\appendix
\noindent{\huge \textbf{Appendices}}

\section{Energy estimates for angular velocity}
\label{ss:filter-analytic-estimates}

We provide the analytic estimates used in Section~\ref{ss:Qsim}.

\begin{lemma}[Concavity of ground energy]\label{lem:filter-force}
Let $V\in C^2(\R{\nbar})$ have bounded Hessian and satisfy $\nabla^2V\succeq K\succ0$.  For fixed $h>0$, let $E(\xi)$ denote the ground energy of $\Hcal_\xi=-\frac{h^2}{2}\Delta+V(z)+\xi^\top z$, where $\xi\in\R{\nbar}$.
Then $E$ is twice differentiable and
\begin{align}\label{eq:filter-force}
    0\preceq-\nabla_\xi^2E(\xi)\preceq K^{-1}.
\end{align}
\end{lemma}
\begin{proof}
Because the potential is confining and grows quadratically, the ground energy is isolated and
simple \cite{reed1978iv}.  Because $\xi^\top z$ is infinitesimally bounded
relative to the quadratic form of $\Hcal_0$, there is a common domain for the family $\Hcal_\xi$ and $\Hcal_0$. Analytic form perturbation theory
\cite[Chapter~VII, Sections~4.2 and~4.6]{kato1995perturbation}
therefore gives the differentiability of $E$.
Using the variational principle express $E(\xi)$ as an infimum of affine
functions of $\xi$, we see that $E$ is concave, which shows the first inequality in \eqref{eq:filter-force} .  To prove the other inequality,
we show that $E(\xi)+\xi^\top K^{-1}\xi/2$ is convex.

For $\beta>0$, define $Z_\beta(\xi):=\Tr e^{-\beta\Hcal_\xi}$.
The quadratic lower bound on the potential allows us to compare the new family with a
harmonic oscillator, so $Z_\beta(\xi)<\infty$ and $-\frac1\beta\log Z_\beta(\xi)\longrightarrow E(\xi)$ as $\beta \to \infty$.
For any integer $h\ge1$, set $\delta=\beta/h$ and use the symmetric Trotter approximation $Z_{\beta,h}(\xi):=\Tr\left[\left(e^{-\delta(V+\xi^\top z)/2}e^{\delta h^2\Delta/2}e^{-\delta(V+\xi^\top z)/2}\right)^h\right]$; for sufficiently large $h$, $Z_{\beta,h}(\xi)\to Z_\beta(\xi)$.
The Gaussian heat kernel gives
\begin{align*}
    Z_{\beta,h}(\xi)
    =c_{\beta,h}\int_{(\R{\nbar})^h}
    \exp\left(-\sum_{j=1}^h\frac{\|z_j-z_{j+1}\|^2}{2h^2\delta}
       -\delta\sum_{j=1}^h[V(z_j)+\xi^\top z_j]\right)
       \d z_1\cdots\d z_h,
\end{align*}
where $z_{h+1}=z_1$ and
$c_{\beta,h}=(2\pi h^2\delta)^{-h\nbar/2}$.
Write $V(z)=z^\top Kz/2+U(z)$, where $U$ is convex, and translate
$z_j=y_j-K^{-1}\xi$.  Then
$e^{-\beta\xi^\top K^{-1}\xi/2}Z_{\beta,m}(\xi)$ is the integral
of the exponential of the negative of
\begin{align*}
    \sum_{j=1}^h\frac{\|y_j-y_{j+1}\|^2}{2h^2\delta}
    +\delta\sum_{j=1}^h
      \left[\frac12y_j^\top Ky_j+U(y_j-K^{-1}\xi)\right].
\end{align*}
This expression is jointly convex in $\xi,y_1,\ldots,y_h$.
Pr\'ekopa's theorem~\cite[Theorem~6]{Prekopa} therefore shows that
the integral is log-concave in $\xi$.  Equivalently, $-\frac1\beta\log Z_{\beta,h}(\xi)+\frac12\xi^\top K^{-1}\xi$ is convex. Taking the limit for $h \to \infty$ and then $\beta \to \infty$ proves that
$E(\xi)+\xi^\top K^{-1}\xi/2$ is convex. Combining this with
the concavity of $E$ proves~\eqref{eq:filter-force}.

\end{proof}

\begin{lemma}[Kinetic energy variance]\label{lem:filter-variance}
Let $V\in C^2(\R{\nbar})$ have bounded Hessian and satisfy $\nabla^2V\succeq mI$ for some $m>0$.
Let $\psi$ be the positive ground state of $\Hcal=T+V$ such that $\|\psi\|=1$, where $T=-h^2\Delta/2$ and $h>0$, and set $\rho=\psi^2$.  Then
\begin{align}\label{eq:filter-variance}
    \bigl\|(T-\langle\psi,T\psi\rangle)\psi\bigr\|_{L^2}^2 =\operatorname{Var}_\rho(V) \le\frac{h^2}{4}\mathbb E_\rho [\Delta V].
\end{align}
\end{lemma}
\begin{proof}
In this proof all expectations without subscripts are with respect to $\rho$.
Let $E$ be the ground energy of $\Hcal$. The eigenvalue equation gives $T\psi=(E-V)\psi$, and hence $(T-\langle\psi,T\psi\rangle)\psi=-(V-\mathbb E[V])\psi$.
Taking squared norms proves the equality in~\eqref{eq:filter-variance}.

Set $f=-\log\psi$, so $\rho=e^{-2f}$.
Because the ground state is log-concave, we have $\nabla^2f\succeq0$~\cite{brascamp1976extensions}.  
Define $Y=\|\nabla f\|^2$, $Z=\Delta f$, $A=\mathbb E[\nabla f^\top(\nabla^2f)\nabla f]$, and $B=\mathbb E\|\nabla^2f\|_{\mathrm F}^2$.
Weighted integration by parts yields
\begin{align}\label{eq:filter-weighted-ibp}
    \mathbb E[u\,\partial_i v]=-\mathbb E[v\,\partial_i u] +2\mathbb E[uv\,\partial_i f],\qquad 
    \mathbb E[u\mathcal Lv]=-\mathbb E[\nabla u\cdot\nabla v],
\end{align}
where $\mathcal L:=\Delta-2\nabla f\cdot\nabla$. 
Taking $u=1$, $v=\partial_i f$ and summing over all $i\in[n]$ gives $\mathbb E[Z]=2\mathbb E[Y]$.  
Moreover, $\mathbb E[Z^2-2ZY]=-\mathbb E[\nabla Z\cdot\nabla f]$ and $B=-\mathbb E[\nabla f\cdot\nabla Z]+2A$.
Consequently, $\operatorname{Var}_\rho(Y-Z) =\operatorname{Var}_\rho(Y)+B-2A$.

For a strictly positive probability density $\eta$ with $\nabla^2(-\log\eta)\succ0$, the Brascamp--Lieb inequality~\cite{brascamp1976extensions,Carlen2013} states that $\operatorname{Var}_\eta(u) \le\mathbb E_\eta [\nabla u^\top(\nabla^2(-\log\eta))^{-1}\nabla u]$. 
Because we have not established that $\nabla^2(-\log \rho)$ is invertible (we only established $\nabla^2 f\succeq0$), we cannot directly use $\eta = \rho = \psi^2$ to invoke the Brascamp--Lieb inequality.
Instead, we apply the Brascamp--Lieb inequality to $u=Y$ and an $\epsilon$-regularized density
$\rho_\epsilon\propto e^{-2f-\epsilon\|z\|^2}$ for some $\epsilon>0$; this guarantees that $-\log(\rho_\epsilon)$ has positive definite Hessian.
Because $\nabla Y=2(\nabla^2f)\nabla f$, the matrix inequality
$D(D+\epsilon I)^{-1}D\preceq D$ for $D\succeq0$ gives
\begin{align*}
    \operatorname{Var}_{\rho_\epsilon}(Y)
       \le 2\mathbb E_{\rho_\epsilon}
       [\nabla f^\top(\nabla^2f)
         (\nabla^2f+\epsilon I)^{-1}(\nabla^2f)\nabla f] \le 2\mathbb E_{\rho_\epsilon}
       [\nabla f^\top(\nabla^2f)\nabla f].
\end{align*}
Letting $\epsilon\downarrow0$ gives
$\operatorname{Var}_\rho(Y)\le2A$.
Together with the estimate on $\operatorname{Var}_\rho(Y-Z)$, this implies
$\operatorname{Var}_\rho(Y-Z)\le B$.
The ground-state equation can also be written as $V-E=\frac{h^2}{2}(Y-Z)$ and $\nabla V=-\frac{h^2}{2}\mathcal L\nabla f$.
Applying~\eqref{eq:filter-weighted-ibp} componentwise gives $B=-\mathbb E[\nabla f\cdot\mathcal L\nabla f]=\frac2{h^2}\mathbb E[\nabla f\cdot\nabla V] =\frac1{h^2}\mathbb E\Delta V$.
It follows that $\operatorname{Var}_\rho(V) =\frac{h^4}{4}\operatorname{Var}_\rho(Y-Z)\le\frac{h^4}{4}B=\frac{h^2}{4}\mathbb E_\rho\Delta V$.
This proves the second part of the identity~\eqref{eq:filter-variance}.

\end{proof}

\section{Strongly convex Schr\"odinger gap}
\label{ss:convex-shrodinger-gap}

\begin{lemma}[Uniformly convex Schr\"odinger gap]
    Let $d \ge 1$ be an integer. Let $W\in C^2(\Rmbb^{d})$ be confining (i.e., $\|W\| \to \infty$ as $\|x\| \to \infty$) and suppose $\nabla^2W\succeq mI$ for some $m>0$.  
    Then, for every $h> 0$, the Schr\"odinger operator $ H =-\frac{h^2}{2}\Delta+W$ has a spectral gap lower bounded by $h\sqrt m$.
\end{lemma}

\begin{proof}
    The proof is a straightforward extension of the strong-convexity gap theorem by Yau~\cite{yau2009estimate}. Yau proved this result in a bounded smooth convex domain $\Omega$, and we will generalize it to the full Euclidean space $\mathbb{R}^d$.
    Without loss of generality, we may assume $W$ is non-negative. On a bounded smooth convex domain $\Omega$, Yau~\cite[Theorems~1.1 and~1.2]{yau2009estimate} proves that that if $\nabla^2U\succeq cI$, then the first two Dirichlet eigenvalues (i.e., the eigenvalues of the operator over $\Omega$ subject to a Dirichlet boundary condition)  of $-\Delta+U$ satisfy
\begin{align}\label{eq:Yau-bounded-gap}
    \lambda_1^\Omega-\lambda_0^\Omega \ge \beta\sqrt c\qquad \text{for~every}\quad 0 < \beta < \sqrt{2}.
\end{align}
Apply~\eqref{eq:Yau-bounded-gap} on the ball $B_R$ to $\frac{2}{h^2}H=-\Delta+\frac{2}{h^2}W$, subject to Dirichlet boundary condition over the $B_R$.
The Hessian of the rescaled potential is bounded below by $2mh^{-2}I$.  
If $W$ is only $C^2$, convolve it with a centered nonnegative mollifier.  Convolution preserves this matrix lower bound, and the mollified potentials converge uniformly on $B_R$.  
The min-max principle passes the Dirichlet eigenvalue bound to $W$. Therefore $\operatorname{gap}_{B_R}\left(\frac{2}{h^2}H\right)\ge\frac{\beta\sqrt{2m}}{h}$.
Because $W$ is confining, if we let $R\to\infty$ then each fixed eigenvalue on $B_R$ converges to the corresponding whole-space eigenvalue. Multiplication by $h^2/2$ and then letting $\beta\uparrow\sqrt2$ give $E_1(H)-E_0(H) \ge h\sqrt m$ as claimed.
\end{proof}

\end{document}